\documentclass{article}

\usepackage{xspace}
\usepackage{euscript}
\usepackage{amsmath,amssymb,amsfonts}
\usepackage[caption=false,font=footnotesize,subrefformat=parens,labelformat=parens]{subfig}	
\usepackage[nice]{nicefrac}
\usepackage{microtype}
\usepackage[margin=1in]{geometry}
\usepackage{esdiff}
\usepackage{cite}
\usepackage{enumitem}
\usepackage{wrapfig}
\usepackage{xargs}
\usepackage[pdftex,dvipsnames]{xcolor}  
\usepackage{comment}
\usepackage{graphicx}
\usepackage{amsthm}
\usepackage[colorlinks]{hyperref}

\def\mc#1{\EuScript{#1}}

\newcommand{\calC}{\ensuremath{\mc{C}}\xspace}

\newcommand{\calV}{\ensuremath{\mc{V}}\xspace}

\newcommand{\calO}{\ensuremath{\mc{O}}\xspace}

\newcommand{\calF}{\ensuremath{\mc{F}}\xspace}

\newcommand{\R}{\mathbb{R}}

\DeclareMathOperator*{\argmin}{argmin}

\newcommand{\ie}{{i.e.}\xspace}

\newcommand{\etal}{\textsl{et~al.}\xspace}

\newcommand{\WLOG}{Without loss of generality\xspace}
\newcommand{\fig}{Figure\xspace}
\newcommand{\Cor}{Corollary\xspace}

\newcommand{\ignore}[1]{}

\def\cl{\mathop{\mathrm{clr}}}
\def\clo{\mathop{\mathrm{cl}}}
\def\interior{\mathop{\mathrm{int}}}

\newcommand{\dist}[1]{\ensuremath{ \lVert #1 \rVert }}

\def\eps{\varepsilon}

\newcommand{\OPT}{\ensuremath{d^*}\xspace}
\newcommand{\dvd}{\ensuremath{\tilde{\calV}}\xspace}

\newcommand{\set}[1]{\ensuremath{\{ #1\}}}

\newcounter{osCounter}
\newcommandx{\oren}[2][1=]{\todo[linecolor=red,backgroundcolor=red!25,bordercolor=red,#1]{OS [\arabic{osCounter}]:  #2}
\stepcounter{osCounter}\xspace}
\newcounter{kfCounter}
\newcommandx{\kyle}[2][1=]{\todo[linecolor=blue,backgroundcolor=blue!25,bordercolor=blue,#1]{KF [\arabic{kfCounter}]:  #2}
\stepcounter{kfCounter}\xspace}
\newcounter{paCounter}
\newcommandx{\pankaj}[2][1=]{\todo[linecolor=OliveGreen,backgroundcolor=OliveGreen!25,bordercolor=OliveGreen, #1]{PA [\arabic{paCounter}]:  #2}
\stepcounter{paCounter}\xspace}

\newtheorem{theorem}{Theorem}[section]
\newtheorem{corollary}[theorem]{Corollary}
\newtheorem{lemma}[theorem]{Lemma}

\begin{document}
%


\title{An Efficient Algorithm for Computing
High-Quality Paths amid Polygonal Obstacles%
\thanks{A preliminary version of this work appear in the Proceedings of the 27th Annual ACM-SIAM
Symposium on Discrete Algorithms.
Most of this work was done while O. Salzman was a student at Tel Aviv University.
Work by P.K. Agarwal and K. Fox was supported in part by NSF under grants CCF-09-40671,
CCF-10-12254, CCF-11-61359, IIS-14-08846, and CCF-15-13816, and by Grant 2012/229 from the
U.S.-Israel Binational Science Foundation.
Work by O. Salzman was supported in part by the Israel Science Foundation (grant no.1102/11),
by the German-Israeli Foundation (grant no. 1150-82.6/2011), by the Hermann Minkowski--Minerva
Center for Geometry at Tel Aviv University and by the National Science Foundation IIS (\#1409003),
Toyota Motor Engineering \& Manufacturing (TEMA), and the Office of Naval Research.
}
}

\author{Pankaj K. Agarwal%
\thanks{Duke University, {pankaj@cs.duke.edu}} \and
Kyle Fox%
\thanks{Duke University, {kylefox@cs.duke.edu}} \and
Oren Salzman%
\thanks{Carnegie Mellon University, osalzman@andrew.cmu.edu}}


\maketitle

\begin{abstract}
We study a path-planning problem amid a set~$\calO$ of obstacles in~$\R^2$, in which 
we wish to
compute a short path between two points while also maintaining a high clearance from~$\calO$;
the clearance of a point is its distance from a nearest obstacle in~$\calO$.
Specifically, the problem asks for a path minimizing the reciprocal of the clearance integrated
over the length of the path.
We present the first polynomial-time approximation scheme for this problem.
Let~$n$ be the total number of obstacle vertices and let~$\eps \in (0,1]$.
Our algorithm computes in time 
$O(\frac{n^2}{\eps^2} \log \frac{n}{\eps})$
a path of total cost at most~$(1+\eps)$ times the cost of the optimal path.
\end{abstract}


\section{Introduction}
\label{sec:intro}
\noindent
\paragraph{Motivation.}
Robot motion planning deals with planning a collision-free path for a  moving object in an environment cluttered with obstacles~\cite{CBHKKLT05}.
It has applications in diverse domains such as surgical planning and computational biology.
Typically, a \emph{high-quality} path is desired where quality can be measured in terms of
path length, clearance (distance from nearest obstacle at any given time), or smoothness, to
mention a few criteria.

\noindent
\paragraph{Problem statement.}
Let~$\calO$ be a set of polygonal obstacles in the plane, 
consisting of~$n$ vertices in total.
A path $\gamma$ for a point robot moving in the plane 
is a continuous function 
$\gamma : [0,1] \rightarrow \R^2$.
Let $\dist{pq}$ denote the Euclidean distance between two points $p,q$.
The \emph{clearance} of a point~$p$,
denoted by $\cl(p):= \min_{o \in \calO} \dist{po}$,
is the minimal Euclidean distance between $p$ and an obstacle
($\cl(p) = 0$ if $p$ lies in an obstacle).

We use the following cost function, as defined by Wein \etal~\cite{WBH08},
that takes both the length and the clearance of a path $\gamma$ into account:
\begin{equation}
  \mu(\gamma):= \int_\gamma \frac{1}{\cl(\gamma(\tau))}  \mathrm{d}\tau.
	\label{eq:length}
\end{equation}
This criteria is useful in many situations because we wish to find a short path that does not pass too close to the obstacles due to safety requirements.
For two points $p,q \in \R^2$, let~$\pi(p, q)$ be the minimal cost\footnote{Wein \etal assume the minimal-cost path exists. One can formally prove its existence by taking the limit of paths approaching the infimum cost.} of any path between~$p$ and~$q$.

The \emph{(approximate) minimal-cost path problem} is defined as follows:
Given 
the set of obstacles~$\calO$ in~$\R^2$,
a real number $\varepsilon \in (0,1]$,
a start position~$s$ and a target position~$t$,
compute a path between~$s$ and $t$ with cost at most~$(1 + \eps)\cdot \pi(s,t)$.

\noindent
\paragraph{Related work.}
There is extensive work in robotics and computational geometry on computing shortest
collision-free paths for a point moving amid a set of planar obstacles, and by now optimal $O(n
\log n)$ algorithms are known; see Mitchell~\cite{M04} for a survey and~\cite{CW13,HSY13} for
recent results.
There is also work on computing paths with the minimum number of links~\cite{MRW92}.
A drawback of these paths is that they may touch obstacle boundaries and therefore their clearance may be zero.
Conversely, if maximizing the distance from the obstacles is the optimization criteria, 
then the path can be computed 
by constructing a maximum spanning tree
in the Voronoi diagram of the obstacles (see {\'{O}}'D{\'{u}}nlaing and Yap~\cite{OY85}).
Wein \etal~\cite{WBH07} considered the problem of computing shortest paths that have clearance at least~$\delta$ for some parameter~$\delta$.
However, this measure does not quantify the tradeoff between the length and the clearance, and the
optimal path may be very long.
Wein \etal~\cite{WBH08}  suggested the cost function defined in equation~\eqref{eq:length} to
balance find a between
minimizing the path length and maximizing its clearance.
They devise an approximation algorithm to compute near-optimal paths under this metric for a point robot moving amidst polygonal obstacles in the plane.
Their approximation algorithm runs in time polynomial in $\frac{1}{\varepsilon}$, $n$,
and~$\Lambda$ where $\varepsilon$ is the maximal additive error, $n$ is the number of obstacle
vertices,
and~$\Lambda$ is (roughly speaking) the total cost of the edges in the Voronoi diagram of the
obstacles; for the exact definition of~$\Lambda$, see~\cite{WBH08}.
Note that the running time of their algorithm is exponential in the worst-case, because the value
of~$\Lambda$ may be exponential as a function of the input size.
We are not aware of any polynomial-time approximation algorithm for this problem.
It is not known whether the problem of computing the optimal path is NP-hard.

The problem of computing shortest paths amid polyhedral obstacles in~$\R^3$ is NP-hard~\cite{CR87}, and a few heuristics have been proposed in the context of sampling-based motion planning in high dimensions (a widely used approach in practice~\cite{CBHKKLT05}) to compute a short path that has some clearance; see, e.g.,~\cite{SMA01}.

Several other bicriteria measures have been proposed in the context of path planning amid
obstacles in~$\R^2$, which combine the length of the path with curvature, the number of links in the
path, the visibility of the path, etc. (see e.g.~\cite{CDK01,AW00,LMA13} and references therein).
We also note a recent work by Cohen \etal~\cite{CFMNSV15}, which is in some sense dual to the problem studied here:
Given a point set~$P$ and a path~$\gamma$, they define the cost of~$\gamma$ to be the integral of
clearance along the path, and the goal is to compute a minimal-cost path between two given points.
They present an approximation algorithm whose running time is near-linear in the number of points.


\noindent
\paragraph{Our contribution.}
We present an algorithm
 that given~$\calO,s,t$ and~$\eps \in (0,1]$, computes in time $
O\left(
	\frac{n^2}{\eps^2} \log \frac{n}{\eps}
\right)
$
a path from~$s$ to~$t$ whose cost is at most~$(1+\eps) \pi(s,t)$.

As in~\cite{WBH08}, our algorithm is based on sampling, \ie, it employs a weighted geometric
graph~$G = (V,E)$ with~$V \subset \R^2$ and~$s,t \in V$ and computes a minimal-cost path in~$G$ from~$s$ to~$t$.
However, we prove a number of useful properties of optimal paths that enable us to sample much
fewer points and construct a graph of size~$O(\frac{n^2}{\eps^2} \log\frac{n}{\eps})$.

We first compute the Voronoi diagram~$\calV$ of~$\calO$ and then refine each Voronoi cell into
constant-size cells.
We refer to the latter as the \emph{refined Voronoi diagram} of~$\calO$ and denote it by~$\dvd$.
We prove in Section~\ref{sec:prelim} the existence of a path~$\gamma$ from~$s$ to~$t$ whose
cost is~$O(\pi(s,t))$ and that has the following useful properties:
(i)~for every cell~$T\in \dvd$,~$\gamma \cap \text{int}(T)$ is a connected subpath and the
clearances of all points in this subpath are the same; we describe these subpaths as
\emph{well-behaved};
(ii)~for every edge~$e \in \dvd$, there are~$O(1)$ points, called \emph{anchor points}, that
depend only on the two cells incident to~$e$ with the property that either $\gamma$ intersects~$e$
transversally (\ie, $\gamma \cap e$ is a single point) or the endpoints of $\gamma \cap e$ are
anchor points.
We use anchor points to propose a simple~$O(n)$-approximation algorithm
(Section~\ref{subsec:O(n)-apx}).
We then use anchor points and the existence of well-behaved paths to choose a set of~$O(n)$ sample
points on each edge of~$\dvd$ and construct a planar graph~$G$ with $O(n^2)$ vertices and edges so
that the optimal path in $G$ from~$s$ to~$t$ has cost~$O(\pi(s,t))$
(Section~\ref{subsec:constant}).

Finally, we prove additional properties of optimal paths to construct the final graph with
$O(\frac{n^2}{\eps^2} \log \frac{n}{\eps})$ edges (Section~\ref{subsec:ptas}).
Unlike Wein \etal~\cite{WBH08}, we do not connect every pair of sample points on the boundary of a
cell~$T$ of~$\dvd$.
Instead, we construct a small size spanner within~$T$ which ensures that the
number of edges in the graph is only~$O(\frac{n^2}{\eps^2} \log \frac{n}{\eps})$ and
not~$O(\frac{n^3}{\eps^2})$.

\section{Preliminaries}
\label{sec:alg_back}
Recall that~$\calO$ is a set of polygonal obstacles in the plane 
consisting of~$n$ vertices in total.
We refer to the edges and vertices of~$\calO$ as its \emph{features}.
Given a point~$p$ and a feature~$o$, let~$\psi_o(p)$ be the closest point to~$p$ on~$o$
so that~$\dist{po} = \dist{p\psi_o(p)}$.
If a path~$\gamma$ contains two points~$p$ and~$q$, 
we let~$\gamma[p,q]$ denote the subpath of~$\gamma$
between~$p$ and~$q$.
Let $\calF = \clo(\R^2 \setminus \calO)$ denote the \emph{free space}.
We assume in this paper that the free space $\calF$ is bounded.
This assumption can be enforced by placing a sufficiently large bounding box around $\calO$ and
the points~$s$ and~$t$.

\paragraph{Voronoi diagram and its refinement.}
The \emph{Voronoi cell} of a polygon feature~$o$, denoted by~$\calV(o)$, is the set of points
in~$\calF$ for which~$o$ is a closest feature of~$\calO$.
Cell $\calV(o)$ is star shaped, and the interiors of Voronoi cells of two different features are
disjoint.
The \emph{Voronoi diagram} of features of~$\calO$, denoted by~$\calV$, is the planar subdivision
of~$\calF$ induced by the Voronoi cells of features in~$\calO$.
The edge between the Voronoi cells of a vertex and an edge feature is a parabolic arc, and between
two vertex or two edge features, it is a line segment.
See Figure~\subref*{fig:voronoi}.
The Voronoi diagram has total complexity~$O(n)$.
See~\cite{AKL13} for details.

\begin{figure}[t]
  \centering
  \subfloat
   [\sf Voronoi Diagram]
   { 
   	\includegraphics[width=0.45\textwidth]{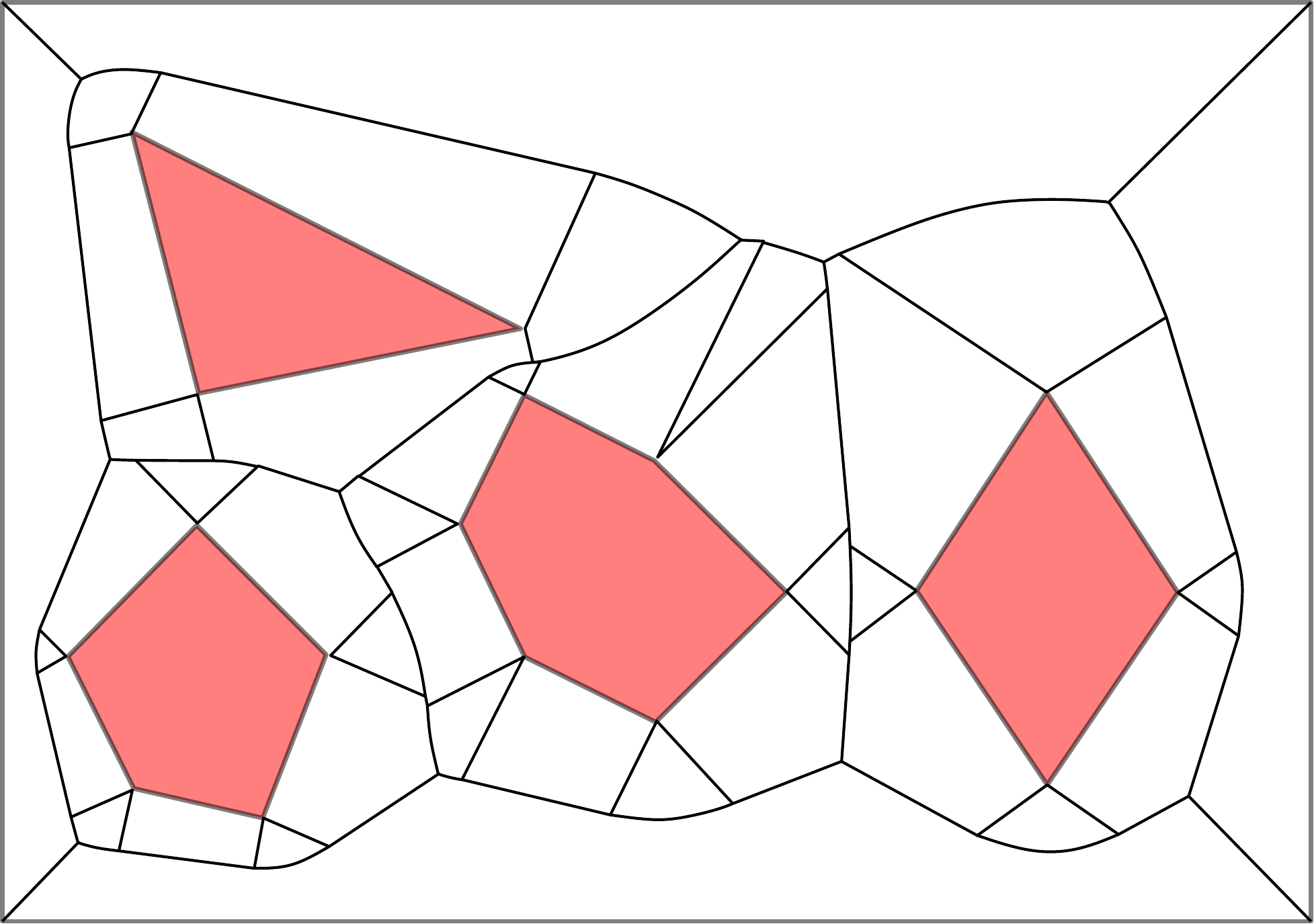}
   	\label{fig:voronoi}
   }
   \hspace{3mm}
  \subfloat
   [\sf Refined Voronoi Diagram]
   {
   	\includegraphics[width=0.45\textwidth]{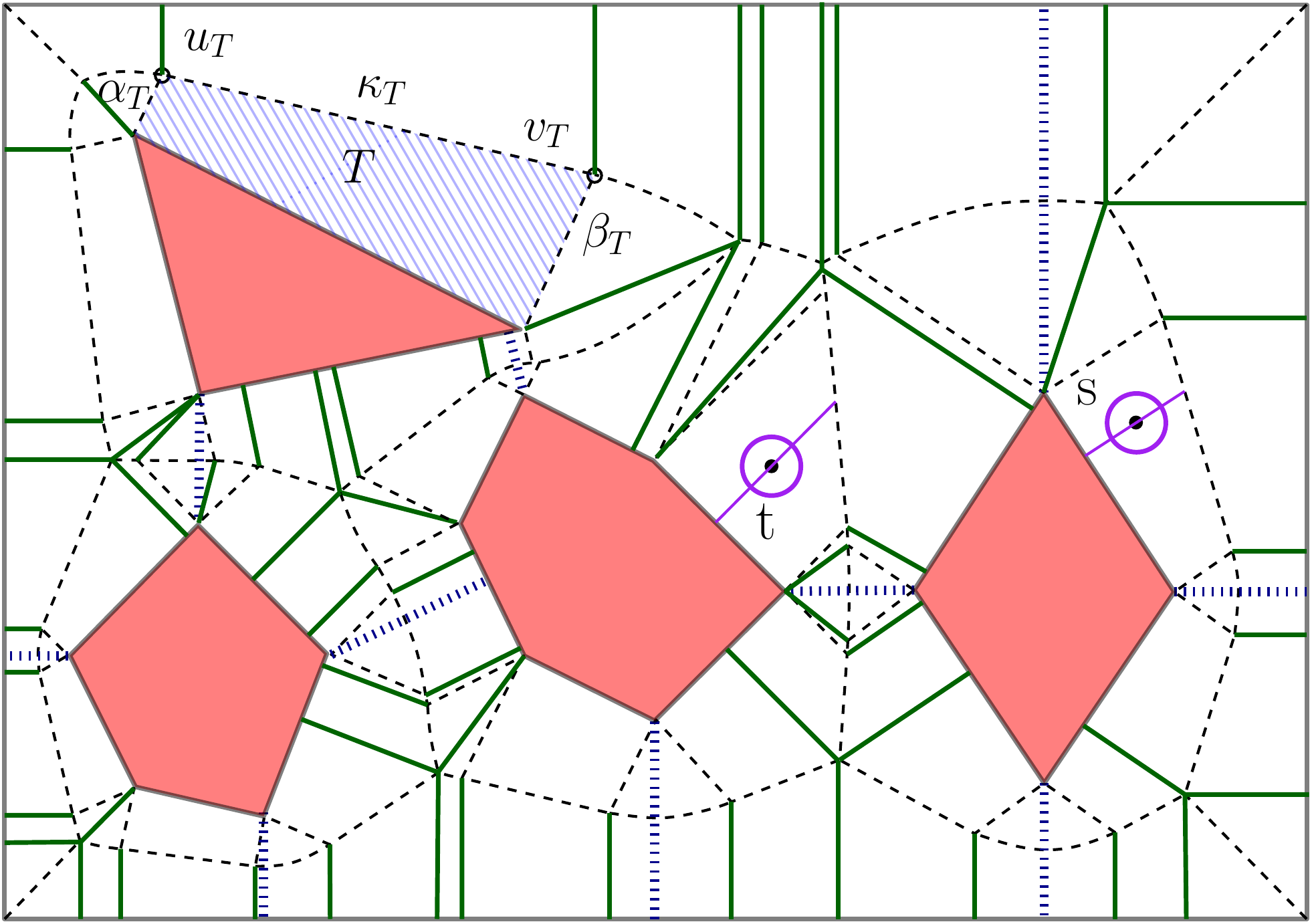}
    \label{fig:refined_voronoi}
   }
	\caption{\sf The Voronoi diagram and the refined Voronoi diagram of a set of obstacles (dark red).
				(a)~Voronoi edges are depicted by solid black lines.
				(b)~Voronoi edges are depicted by dashed black lines.
				Green solid lines and blue dotted lines represent
        edges of type~(i) and type~(ii), respectively.
				A representative cell~$T$ is depicted in light blue.}
  \label{fig:voronoi_all}
\end{figure}

For any obstacle feature~$o$ and for any point~$x$ along any edge on~$\partial \calV(o)$, the
function~$\dist{x\psi_o(x)}$ is convex.
We construct the \emph{refined Voronoi diagram}~\dvd by adding the following edges to each Voronoi
cell~$\calV(o)$ and refining it into constant-size cells:
\begin{enumerate}[label={\normalfont (\roman*)}]
  \item
    the line segments~$x \psi_o(x)$ between each obstacle feature~$o$ and every vertex~$x$
on~$\calV(o)$ and
\item
  for each edge~$e$ of~$\calV(o)$, the line segment~$x \psi_o(x)$, where~$x \in e$ is the point
that minimizes~$\dist{x \psi_o(x)}$.
\end{enumerate}
We also add a line segment from the obstacle feature~$o$ closest to~$s$ (resp. $t$) that initially
follows $\psi_o(s) s$ (resp. $\psi_o(t) t$) and ends at the first Voronoi edge it intersects.
Note that some edges of type (i) may already be present in the Voronoi diagram~$\calV$.
We say that an edge in~\dvd is an \emph{internal edge} if it separates two cells incident to the
same polygon.
Other edges are called \emph{external edges}.

Clearly, the complexity of~\dvd is $O(n)$.
Moreover, each cell~$T$ in~\dvd is incident to a single obstacle feature~$o$ and has three
additional edges.
One edge is external, and it is a monotone parabolic arc or line segment.
The other two edges are internal edges on~$T$, and they are both line segments.
For each cell~$T$, let~$\kappa_T$ denote the external edge of~$T$, let~$\alpha_T$ and~${\beta_T}$
denote the shorter and longer internal edges of~$T$ respectively, and let~$u_T$ and~$v_T$ denote
the vertices connecting~$\alpha_T$ and~${\beta_T}$ to~$\kappa_T$ respectively.
See Figure~\subref*{fig:refined_voronoi}.

For any value $c > 0$, the set of points in Voronoi cell~$T$ of clearance~$c$, if nonempty, forms
a connected arc~$\eta$ which is a circular arc centered at~$o$ if~$o$ is a vertex and a line
segment parallel to~$o$ if~$o$ is an edge.
One endpoint of~$\eta$ lies on~$\beta_T$ and the other on~$\alpha_T$ or~$\kappa_T$.

\noindent
\paragraph{Properties of optimal paths.}
We list several properties of our cost function.
For detailed explanations and proofs, the reader is referred to Wein \etal~\cite{WBH08}.
Let~$s = r_s e^{i \theta_s}$ be a start position and~$t  = r_t e^{i \theta_t}$ be a target
position.
\begin{enumerate}[label=(P\arabic*)]
  \item
    Let~$o$ be a point obstacle with~$\calO = \{o\}$, and assume without loss of generality
    that~$o$ lies at the origin and~$0 \!\leq \!\theta_s \leq \!\theta_t \leq \!\pi$.
  The optimal path  between~$s$ and~$t$ (see Figure~\subref*{fig:sp_point_obs}) is a logarithmic spiral centered on~$o$, and its cost is
	\begin{equation}
		\pi(s,t) = \sqrt{(\theta_t - \theta_s)^2 
								+ 
							(\ln r_t - \ln r_s)^2}.
		\label{eq:length_point}
	\end{equation}

\item
    Let~$o$ be a line obstacle with~$\calO = \{o\}$, and assume without loss of generality
    that~$o$ is supported by the line $y=0$, $0 \leq \theta_s \leq \theta_t \leq \pi$, and $r_s =
    r_t = r$.
  The optimal path between~$s$ and~$t$ (see Figure~\subref*{fig:sp_line_obs}) is a circular arc with its center at the origin, and its
  cost\footnote{The original equation describing the cost of the optimal path in the vicinity
    of a line obstacle had the obstacle on~$x =0$, and it contained a minor inaccuracy in its
    calculation. We present the
    correct cost in~\eqref{eq:length_line}.} is
\begin{equation}
  \pi(s,t) = \ln \frac{1 - \cos \theta_t}{\sin \theta_t}
          -
          \ln \frac{1 - \cos \theta_s}{\sin \theta_s}
         = \ln \tan \frac{\theta_t}{2} - \ln \tan \frac{\theta_s}{2}.
  \label{eq:length_line}
\end{equation}

\item
    Let~$o$ be an obstacle with~$\calO = \{o\}$ and $s$ on the line segment between~$\psi_o(t)$
    and~$t$.
  The optimal path between~$s$ and~$t$ (see Figure~\subref*{fig:sp_degenerate}) is a line segment, and its cost is
  \begin{equation}
    \pi(s,t) = \ln \cl(t) - \ln \cl(s).
    \label{eq:length_away}
  \end{equation}

\item
    The minimal-cost path~$\gamma$ between two points~$p$ and~$q$ on an edge~$e$ of~$\calV$ is the piece of~$e$ between~$p$ and~$q$.
	Moreover, there is a closed-form formula describing the cost of~$\gamma$.

\item
  Since each point within a single Voronoi cell is closest to exactly one obstacle feature, we may
  conclude the following:
  Given a set of obstacles, the optimal path connecting $s$ and $t$ consists of a sequence of
  circular arcs, pieces of logarithmic spirals, line segments, and pieces of Voronoi edges.
  Each of member of this sequence begins and ends on an edge or vertex of~$\dvd$ (see
    Figure~\subref*{fig:sp_polygonal}).

\end{enumerate}
\begin{figure}[t]
  \centering
  \subfloat
   [\sf Point obstacle]
   { 
   	\includegraphics[scale=.34]{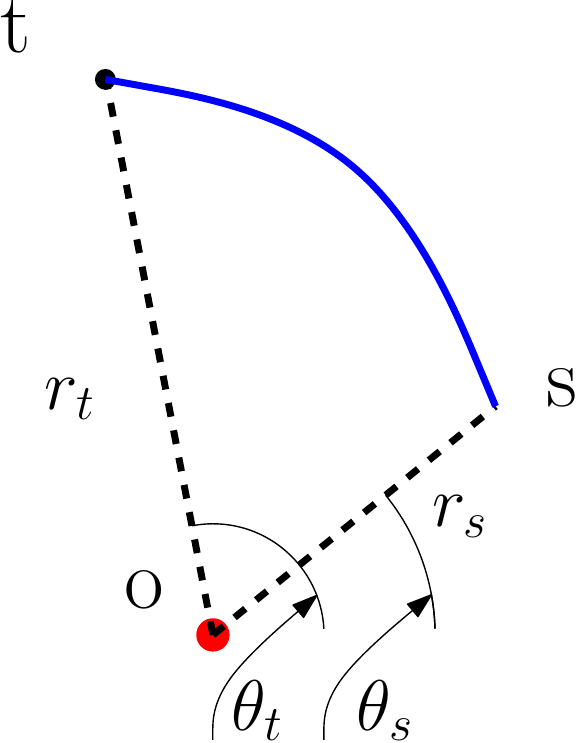}
   	\label{fig:sp_point_obs}
   }
   \hspace{2mm}
  \subfloat
   [\sf Line obstacle]
   { 
   	\includegraphics[scale=.34]{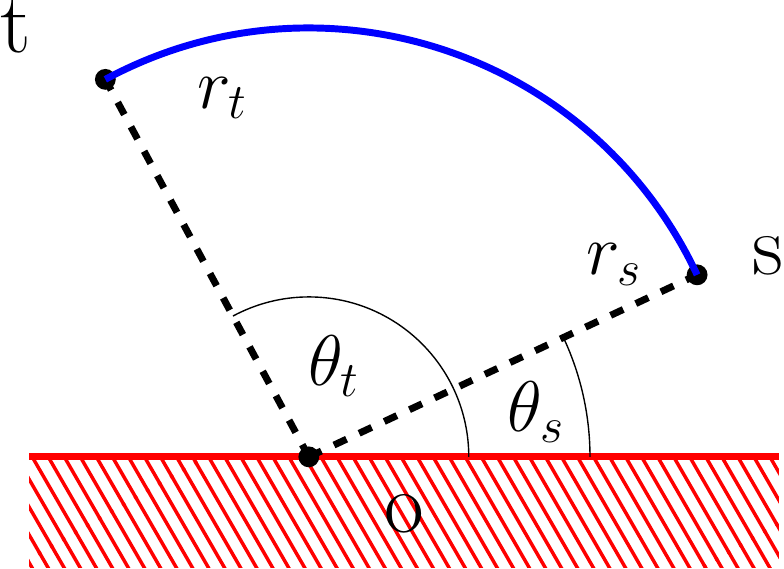}
   	\label{fig:sp_line_obs}
   }
   \hspace{2mm}
  \subfloat
   [\sf Line obstacle (degenerate)]
   { 
    \hspace{.08\textwidth}
   	\includegraphics[scale=.34]{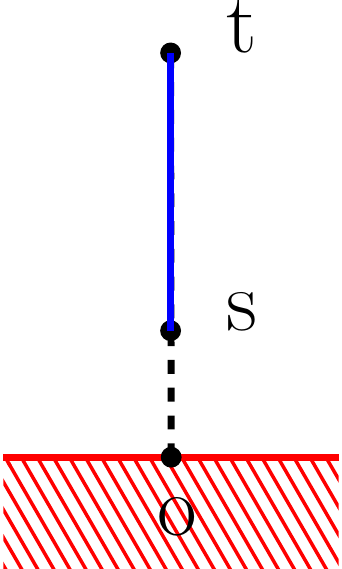}
    \hspace{.08\textwidth}
   	\label{fig:sp_degenerate}
   }\\
  \subfloat
   [\sf Polygonal obstacles]
   { 
   	\includegraphics[scale=.25]{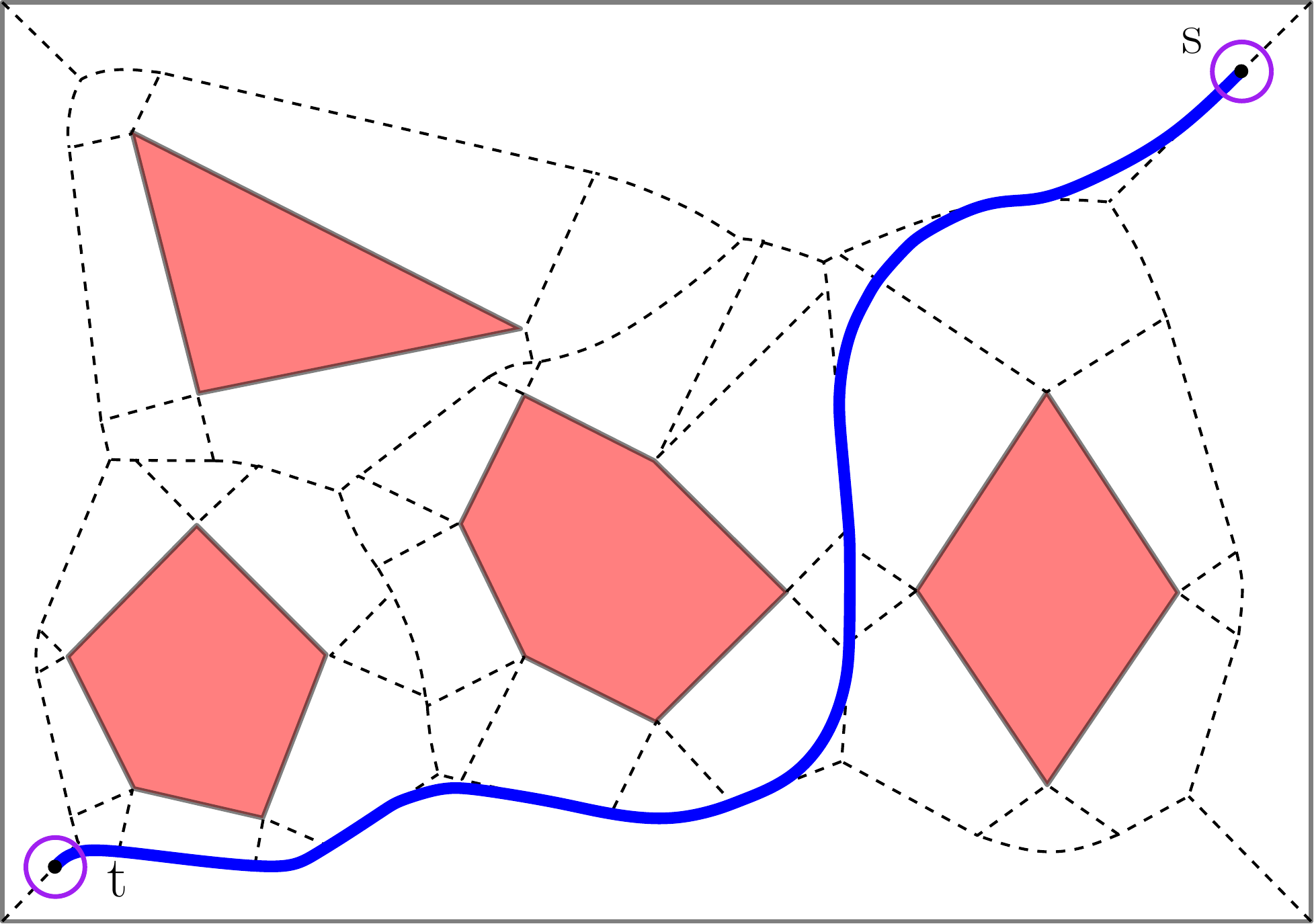}
   	\label{fig:sp_polygonal}
   }
   	\caption{\sf Different examples of optimal paths (blue) among different types obstacles (red).
   				In (d), the Voronoi diagram of the obstacles is depicted by dashed black lines.}
  \label{fig:sp_examples}
\end{figure}

The following lemma follows immediately from~\eqref{eq:length} and~\eqref{eq:length_point}.

\begin{lemma}
\label{lem:cost_observations}
	Let~$p$ and~$q$ be two points 
	such that~$\cl(p) \leq \cl(q)$.
	The following properties hold:
  \begin{enumerate}[label={\normalfont (\roman*)}]
  \item
    We have $\pi(p,q) \geq \ln \frac{\cl(q)}{\cl(p)}$.
If~$p$ and~$q$ lie in the same Voronoi cell of an 
obstacle feature~$o$ and if~$p$ lies on the line segment $q\psi_o(q)$, then 
the bound is tight.\\
\item
  If there is a single point obstacle~$o$ located at the origin, 
$p = r_p e^{i \theta_p}$ 
and 
$q = r_q e^{i \theta_q}$
with~$0 \leq \theta_p \leq \theta_q \leq \pi$, 
then $\pi(p,q) \geq \theta_q - \theta_p$.
If $r_p = r_q$ (namely,~$p$ and~$q$ are equidistant to~$o$),
then the bound is tight.
\end{enumerate}
\end{lemma}

An immediate corollary of Lemma~\ref{lem:cost_observations} is the following:
\begin{corollary}
  \label{cor:clearance_bound}
  Let~$p$ and~$q$ be two points on the boundary of a Voronoi cell~$T$.
  Let~$w$ be another point on the same edge of $T$ as $p$, and let $\cl(p) \leq \cl(w) \leq
  \cl(q)$.
  Then~$\pi(p,w) \leq \pi(p,q)$.
\end{corollary}

\paragraph{Model of computation.}
We are primarily concerned with the combinatorial time complexity of our algorithm.
Therefore, we assume a model of computation that allows us to evaluate basic trigonometric
and algebraic expressions, such as the ones given above, in constant time.
Our model also allows us to find the roots of a constant-degree polynomial in constant time.

\section{Well-behaved Paths}
\label{sec:prelim}
%

Let~$T$ be a cell of~$\dvd$ incident to obstacle feature~$o$, and let~$p$
and~$q$ be two points on~$\partial T$.
\begin{wrapfigure}{r}{0.4\textwidth}
  \begin{center}
    \includegraphics[scale=.34]{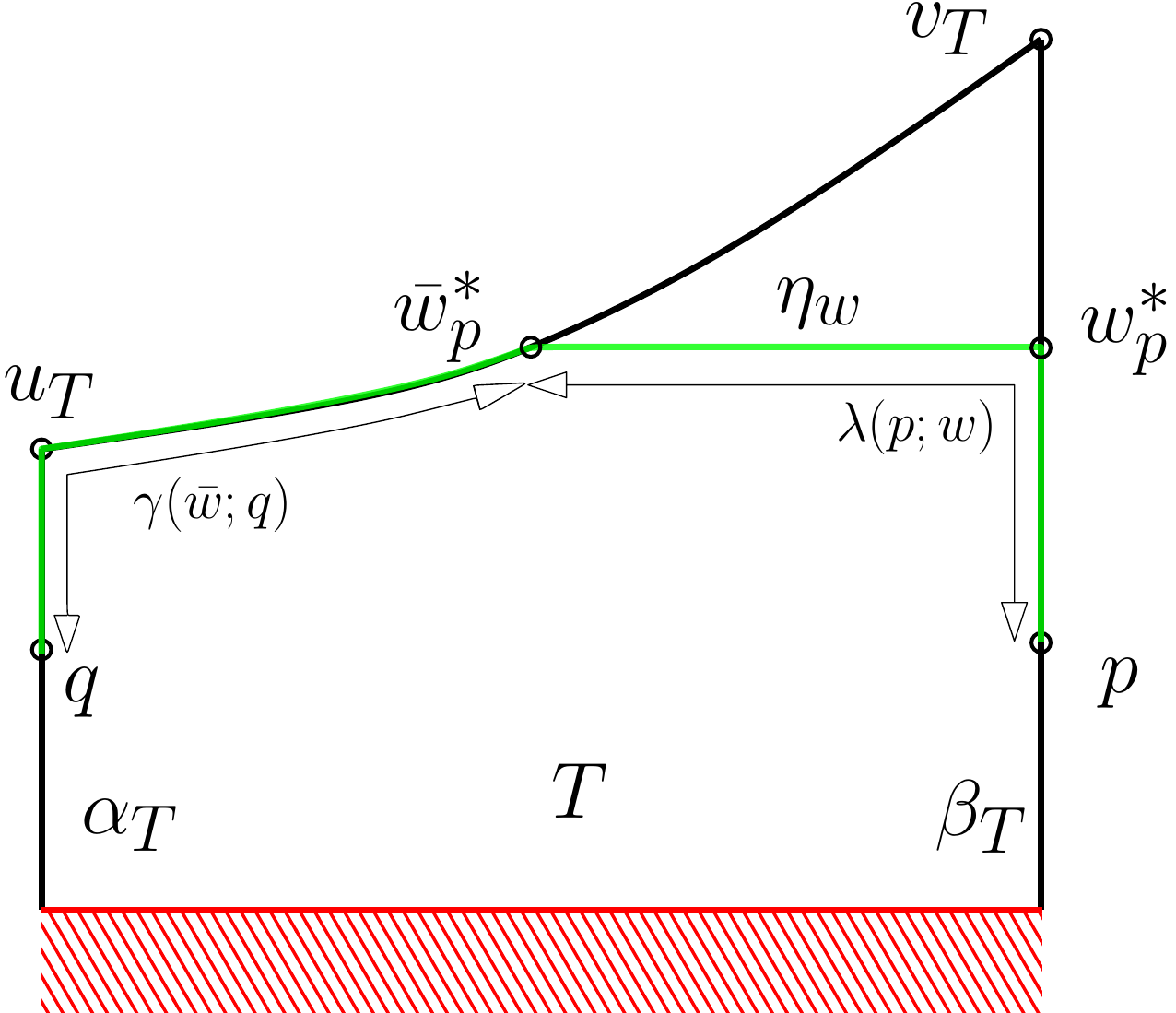}
  \end{center}
  \caption{Components of the well-behaved path~$\gamma(p,q)$.}
  \label{fig:section3_notations}
\end{wrapfigure}
We define a \emph{well-behaved} path between~$p$ and~$q$, denoted by~$\gamma(p,q)$,
whose cost is at most~$11\pi(p,q)$ and that can be computed in~$O(1)$ time.
We first define~$\gamma(p,q)$, then analyze its cost, and finally prove an additional property
of~$\gamma(p,q)$ that allows us to compute it in~$O(1)$ time.

If both~$p$ and~$q$ lie on the same edge of~$\partial T$ or neither of them lies on the
edge~$\beta_T$, then we define~$\gamma(p,q)$ to be the unique path from~$p$ to~$q$ along~$\partial
T$ that does not intersect~$o$.
If one of~$p$ and~$q$, say,~$p$, lies on~$\beta_T$, then~$\gamma(p,q)$ is somewhat more involved,
because the path along~$\partial T$ can be quite expensive.
Instead, we let~$\gamma(p,q)$ enter the interior of~$T$.
For a point~$w \in \beta_T$, let~$\eta_w$ be the maximal path in~$T$ of clearance~$\cl(w)$
beginning at~$w$, i.e., its image is the set of points~$\set{z \in T \mid \cl(z) = \cl(w)}$.

By the discussion in Section~\ref{sec:alg_back},~$\eta_w$ is a line segment or a circular arc
with~$w$ as one of its endpoints.
Let~$\bar{w}$ be the other endpoint of~$\eta_w$.
We define the path
$$\lambda(p;w) = pw \circ \eta_w$$
to be the segment of~$pw$ followed by the arc~$\eta_w$.
We refer to~$w$ as the \emph{anchor point} of~$\lambda(p; w)$.
Let $w^*_p$ be the anchor point on edge $\beta_T$ of clearance greater than~$\cl(p)$ that
minimizes the cost of $\lambda(p;w)$.
Namely,
$$w^*_p = \argmin_{%
\substack{%
w \in \beta_T\\
\cl(w) \geq \cl(p)
}}
\mu(\lambda(p;w)).$$
We now define
\begin{align*}
  \gamma(p,q; w) &= \lambda(p;w) \circ \gamma(\bar{w},q),\\
  \gamma(p,q) &= \gamma(p,q; w^*_p).
\end{align*}
See Figure~\ref{fig:section3_notations}.

The next two lemmas bound the cost of~$\gamma(p,q)$.
\begin{lemma}~
  \label{lem:easy_well-behaved_cost}
  \begin{enumerate}[label={\normalfont (\roman*)}]
    \item
      If~$p$ and~$q$ lie on the same edge of~$\partial T$, then $\mu(\gamma(p,q)) = \pi(p,q)$.
    \item
      If neither~$p$ nor~$q$ lies on~$\beta_T$, then $\mu(\gamma(p,q)) \leq 3\pi(p,q)$.
  \end{enumerate}
\end{lemma}
\begin{proof}~
  \begin{enumerate}[label=(\roman*)]
    \item
      If~$p$ and~$q$ lie on the same edge~$e$ of~$\partial T$, then~$\gamma(p,q) \subseteq e$, and
      by (P4),~$\gamma(p,q)$ is the optimal path between~$p$ and~$q$.
      Hence, the claim follows.
    \item
      Suppose~$p \in \alpha_T$ and~$q \in \kappa_T$.
      Path~$\gamma(p,q)$ travels along~$\alpha_T$ from~$p$ to~$u_T$, and then along~$\kappa_T$
      from~$u_T$ to~$q$.
      By \Cor~\ref{cor:clearance_bound} and the fact that $u_T$ is the lowest clearance point
      on~$\kappa_T$, we have~$\pi(p,u_T) \leq \pi(p,q)$.
      By the triangle inequality, we have that
      $$\pi(u_T,q) \leq \pi(u_T,p) + \pi(p,q) \leq 2 \pi(p,q).$$
      Finally,
      $$\mu(\gamma(p,q)) = \pi(p,u_T) + \pi(u_T,q) \leq 3 \pi(p,q).$$
  \end{enumerate}
\end{proof}

\begin{lemma}
  \label{lem:hard_well-behaved_cost}
  If~$p \in \beta_T$ and~$q \notin \beta_T$, then~$\mu(\gamma(p,q)) \leq 11\pi(p,q)$.
\end{lemma}
\begin{proof}
  Let~$w$ be any point of~$\beta_T$ such that $\cl(w) \geq \cl(p)$.
  We begin by proving $\mu(\gamma(p,q;w)) \leq 4\mu(\lambda(p; w)) + 3\pi(p,q)$.
  Later, we will show $\mu(\lambda(p;w^*_p)) \leq 2\pi(p,q)$, proving the lemma.

  To prove the first claim, we consider different cases depending on the edges of~$\partial T$
  that contain~$\bar{w}$ and~$q$.
  See Figure~\ref{fig:hard_well-behaved_cost}.
  \begin{figure*}[t]
    \begin{center}
      \subfloat
      [\sf Case 1: $q \in \kappa_T$]
      { 
      \includegraphics[scale=.34]{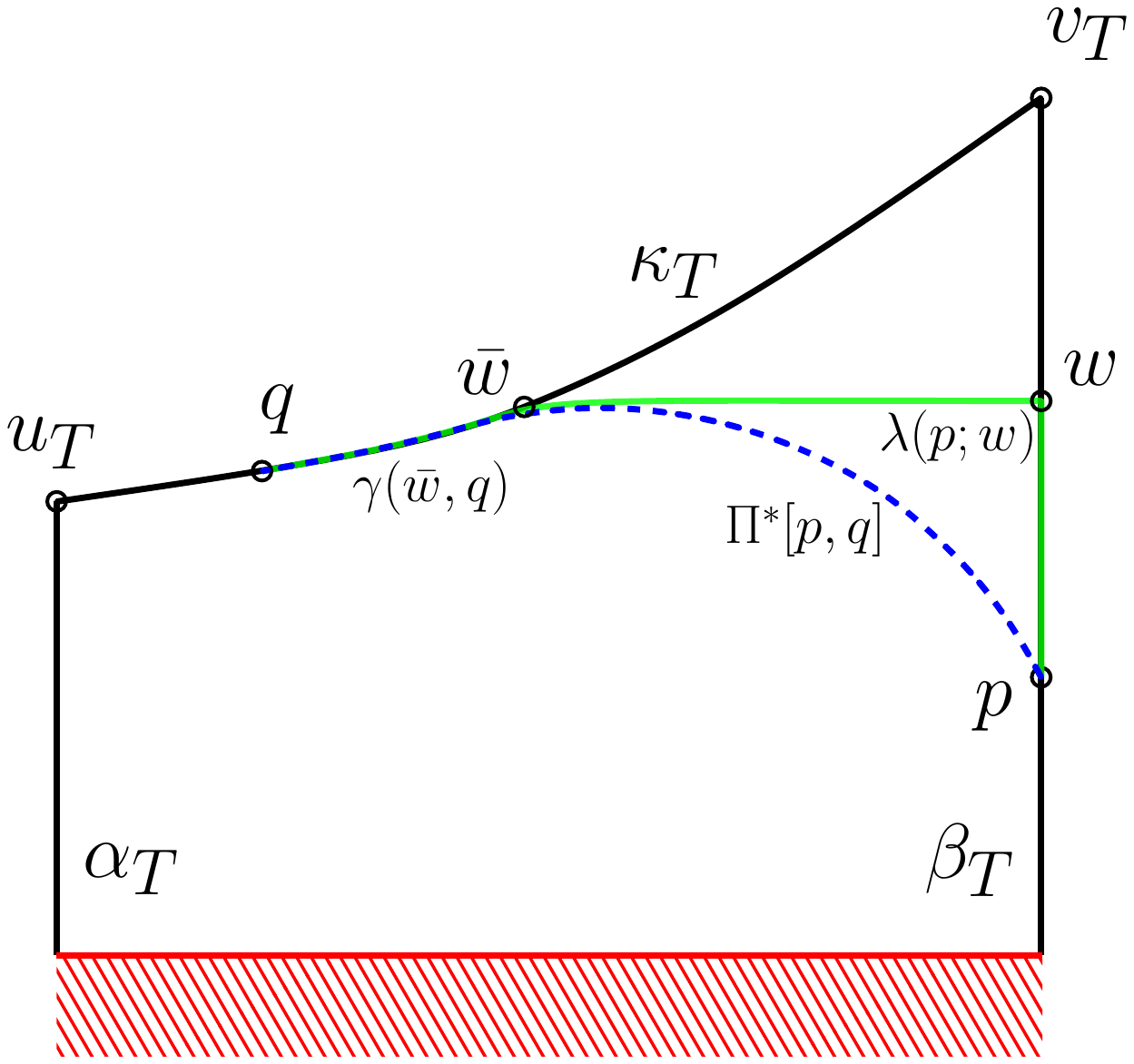}
      \label{fig:hard_i}
      }
      \hspace{0.02\textwidth}
      \subfloat
      [\sf Case 2: $q,\bar{w} \in \alpha_T$]
      {
      \includegraphics[scale=.34]{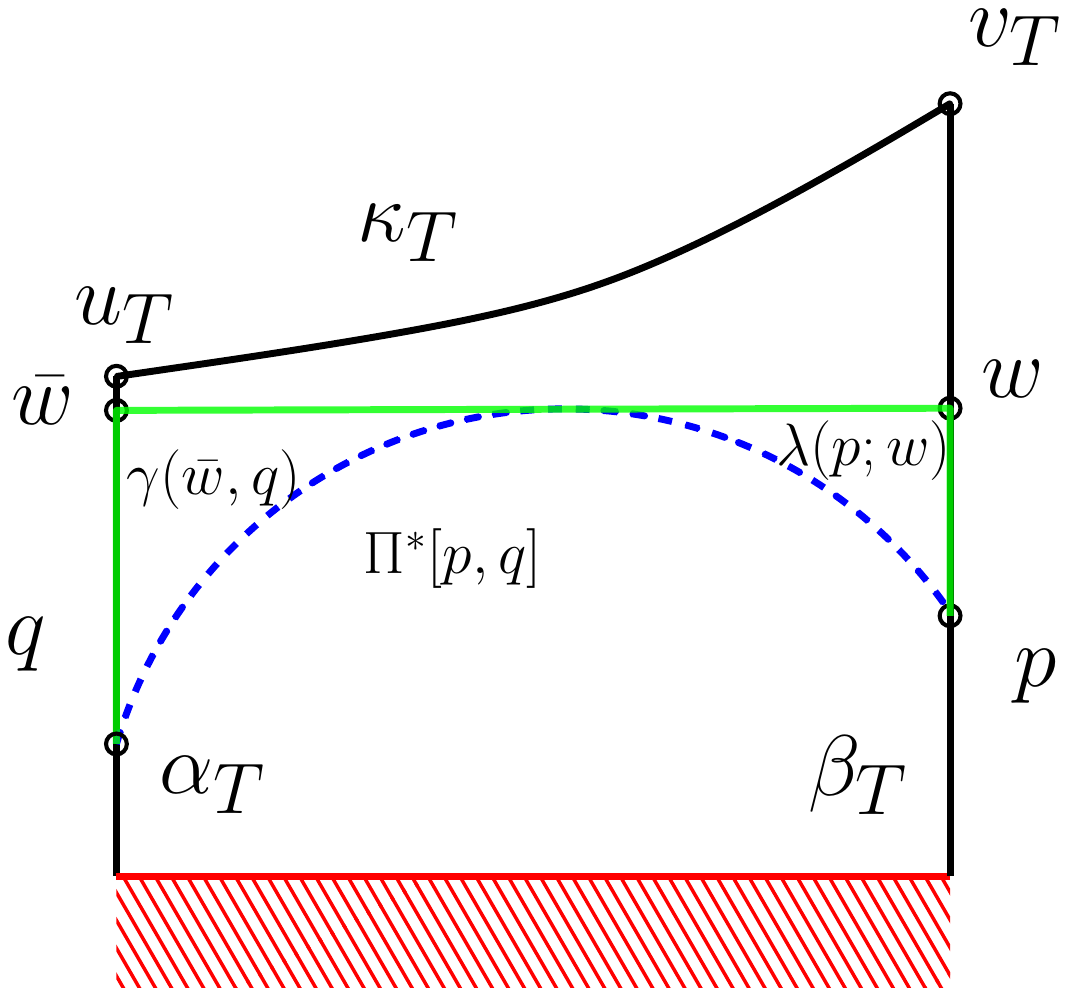}
      \label{fig:hard_ii}
      }
      \hspace{0.02\textwidth}
      \subfloat
      [\sf Case 3: $q \in \alpha_T, \bar{w} \in \kappa_T$]
      {
      \includegraphics[scale=.34]{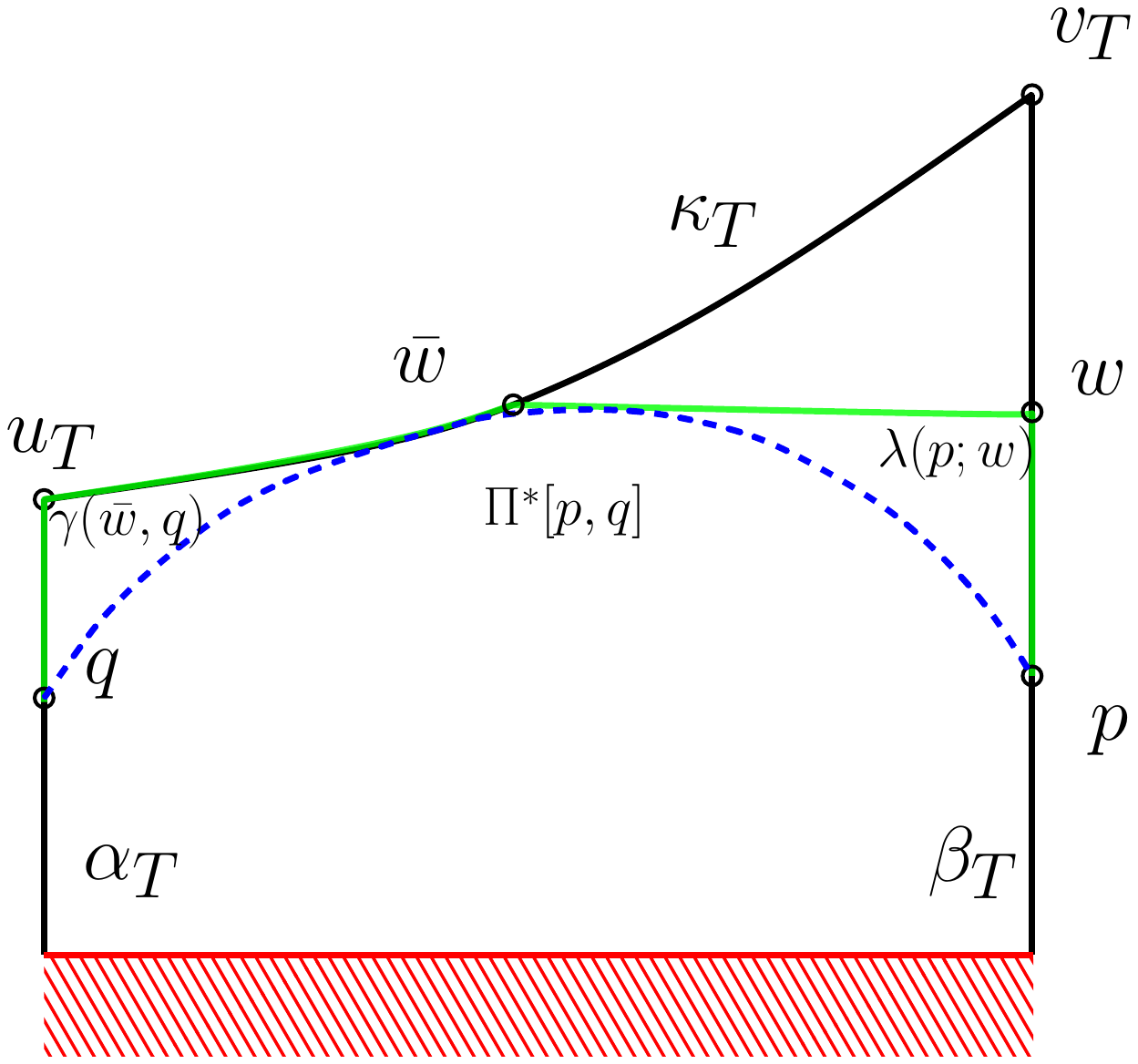}
      \label{fig:hard_iii}
      }
      \caption{\sf Different cases considered in the proof of
      Lemma~\ref{lem:hard_well-behaved_cost}.
      We use $\Pi^*[p,q]$ to denote the minimal-cost path between $p$ and $q$.}
      \label{fig:hard_well-behaved_cost}
    \end{center}
  \end{figure*}
  \begin{enumerate}[align=left]
    \item[Case 1: $q \in \kappa_T$.]
      In this case, $\gamma(\bar{w},q) \subseteq \kappa_T$, and therefore $\mu(\gamma(\bar{w},q))
      = \pi(\bar{w},q)$.
      By the triangle inequality, $\pi(\bar{w},q) \leq \mu(\lambda(p; w)) + \pi(p,q)$.
    \item[Case 2: $q,\bar{w} \in \alpha_T$.]
      In this case, $\gamma(\bar{w},q) \subseteq \alpha_T$, and therefore $\mu(\gamma(\bar{w},q))
      = \pi(\bar{w},q)$.
      Again, $\pi(\bar{w},q) \leq \mu(\lambda(p; w)) + \pi(p,q)$.
    \item[Case 3: $q \in \alpha_T, \bar{w} \in \kappa_T$.]
      In this case, $\gamma(\bar{w},q)$ first travels
      along~$\kappa_T$ from~$\bar{w}$ to~$u_T$ and then along~$\alpha_T$ from~$u_T$ to~$q$.
      Since $\cl(u_T) \leq \cl(w)$,
      $$ \pi(q, u_T) \leq \pi(q,p) + \pi(p, w) \leq \pi(p,q) + \mu(\lambda(p; w))$$
      by Corollary~\ref{cor:clearance_bound}.
      Furthermore, by the triangle inequality,
      \begin{align*}
        \pi(\bar{w}, u_T) &\leq \mu(\lambda(p;w)) + \pi(p,q) + \pi(q,u_T)\\
        &\leq 2\mu(\lambda(p;w)) + 2\pi(p,q).
      \end{align*}
      Hence, $\mu(\gamma(\bar{w},q)) = \pi(\bar{w}, u_T) + \pi(u_T, q) \leq 3\mu(\lambda(p;w)) +
      3\pi(p,q)$.
  \end{enumerate}
  Since $\mu(\gamma(\bar{w},q)) \leq 3\mu(\lambda(p;w)) + 3\pi(p,q)$ in all three cases,
  $\mu(\gamma(p,q; w)) \leq 4\mu(\lambda(p;w)) + 3\pi(p,q)$ as claimed.

  Let~$c^*$ be the maximum clearance of a point on the optimal path between~$p$ and~$q$ (if there
  are multiple optimal paths between~$p$ and~$q$, choose one of them arbitrarily).
  Let~$w \in \beta_T$ be the point of clearance $\min\set{c^*, \cl(v_T)}$.
  We now prove that $\mu(\lambda(p;w)) \leq 2 \pi(p,q)$.

  We first note that $\cl(p) \leq \cl(w) \leq c^*$.
  Therefore, by Corollary~\ref{cor:clearance_bound}, $\mu(pw) = \pi(p,w) \leq \pi(p,q)$.
  Next, we argue that $\mu(\eta_w) \leq \pi(p,q)$.
  Indeed, if~$o$ is a polygon edge, then~$\eta_w$ is the Euclidean shortest path between any pair
  of points on~${\beta_T}$ and one of~$\alpha_T$ or~$\kappa_T$ whose clearance never
  exceeds~$c^*$.
  It also (trivially) has the highest clearance of any such path.
  If~$o$ is a polygon vertex, then~$\eta_w$ spans a shorter angle relative to $o$ than any other
  path whose clearance never exceeds~$c^*$.
  By Lemma~\ref{lem:cost_observations}(ii), the cost of any such path from~${\beta_T}$ to one
  of~$\alpha_T$ or~$\kappa_T$ is at least this angle, and by (P1), the cost of~$\eta_w$ is exactly
  this lower bound.
  Either way, any path between~$p$ and~$q$ also goes between~${\beta_T}$ and one of~$\alpha_T$
  or~$\kappa_T$, so we conclude that~$\mu(\eta_w) \leq \pi(p,q)$.
  Hence, $\mu(\lambda(p;w)) \leq 2\pi(p,q)$.
  In particular, $\mu(\lambda(p;w^*_p)) \leq 2\pi(p,q)$, and $\mu(\gamma(p,q;w^*_p)) \leq
  11\pi(p,q)$. 
\end{proof}

If~$p \in \beta_T$ and~$q \notin \beta_T$, then computing~$\gamma(p,q)$ requires computing the
anchor point~$w$ that minimizes~$\mu(\lambda(p,w))$.
We show that the point~$w^*_p \in \beta_T$ that defines~$\gamma(p,q)$ is either~$p$ itself or a
point that only depends on the geometry of~$\partial T$ and not on~$p$ or~$q$.

\begin{lemma}
\label{lem:good_anchors}
There exist two points~$w_{\alpha_T}$ and $w_{\kappa_T}$ on~$\beta_T$, such that for any~$p \in
\beta_T$, $w^*_p \in \set{p, w_{\alpha_T}, w_{\kappa_T}}$.
Furthermore, these two points can be computed in~$O(1)$ time.
\end{lemma}
\begin{proof}
  There are several cases to consider depending on whether~$o$ is a vertex or an edge, whether the
  point~$\bar{w}^*$ lies on~$\kappa_T$ or~$\alpha_T$, and whether~$\kappa_T$ is a line segment or
  a parabolic arc.
  Depending on the geometry of~$T$, we define~$w_{\alpha_T}$ and~$w_{\kappa_T}$ accordingly.
  In each case, we parameterize the anchor point~$w$ on~$\beta_T$ appropriately and show that
  $w^*_p \in \set{p, w_{\alpha_T}, w_{\kappa_T}}$.
  For simplicity, for a parameterized anchor point $w(t)$, we use~$\eta(t)$, $\lambda(t)$, and
  $\mu(t)$ to denote~$\eta_{w(t)}$, $\lambda(p; w(t))$, and $\mu(\lambda(p;w(t)))$, respectively.
  We now describe each case:
  \begin{enumerate}[wide,itemindent=0pt]
    \item[Case 1: $o$ is a vertex.]
      \WLOG, assume that~$o$ lies at the origin, edges~${\alpha_T}$ and~$\beta_T$ intersect the
      line~$y = 0$ at the origin with angles~$\theta_\alpha$ and $\theta_\beta$ respectively,
      and~$\theta_\beta > \theta_\alpha \geq 0$.
      In this case,~$\eta_w$, the constant-clearance path anchored at~$w \in \beta_T$, is a circular
      arc.
      We consider two cases depending on whether~$\bar{w}^*$ lies on~$\alpha_T$ or~$\kappa_T$.
      See Figure~\ref{fig:good_anchors_1}.
      \begin{figure*}[t]
        \begin{center}
          \subfloat
          [\sf Case 1(a): $\bar{w}^* \in \alpha_T$]
          { 
          \includegraphics[scale=.34]{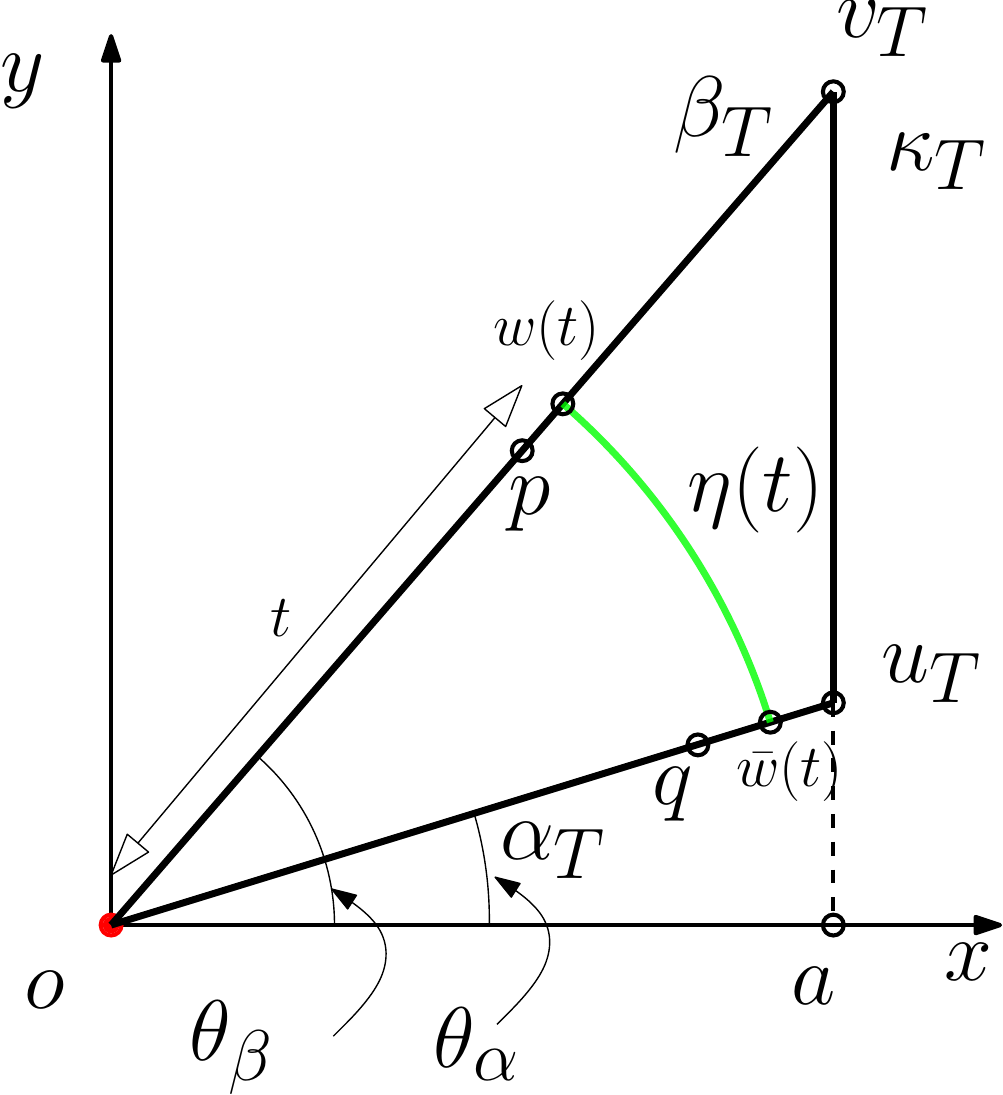}
          \label{fig:good_anchors_1a}
          }
          \hspace{0.03\textwidth}
          \subfloat
          [\sf Case 1(b)(i): $\bar{w}^* \in \kappa_T$; $\kappa_T$ is a line segment]
          {
          \includegraphics[scale=.34]{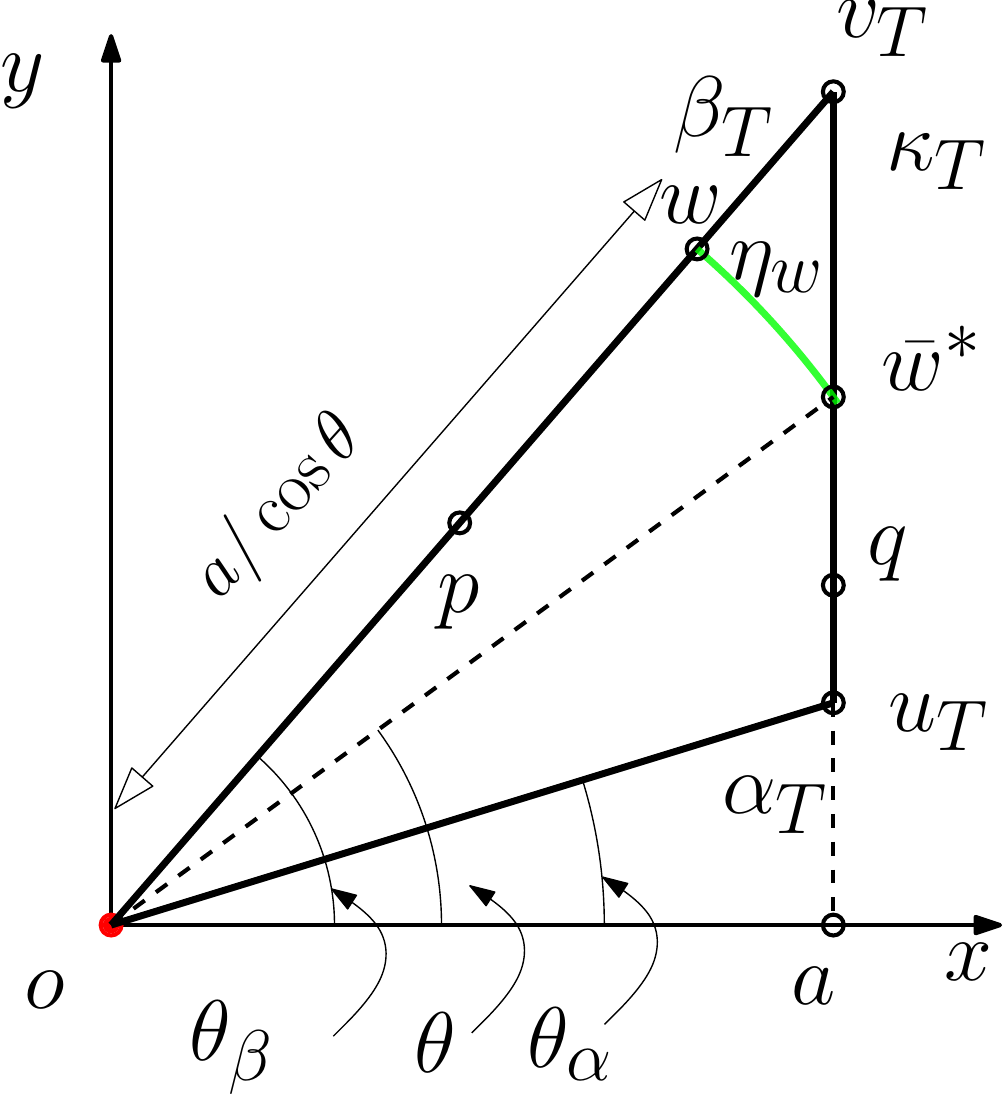}
          \label{fig:good_anchors_1bi}
          }
          \hspace{0.01\textwidth}
          \subfloat
          [\sf Case 1(b)(ii): $\bar{w}^* \in \kappa_T$; $\kappa_T$ is a parabolic arc]
          {
          \includegraphics[scale=.34]{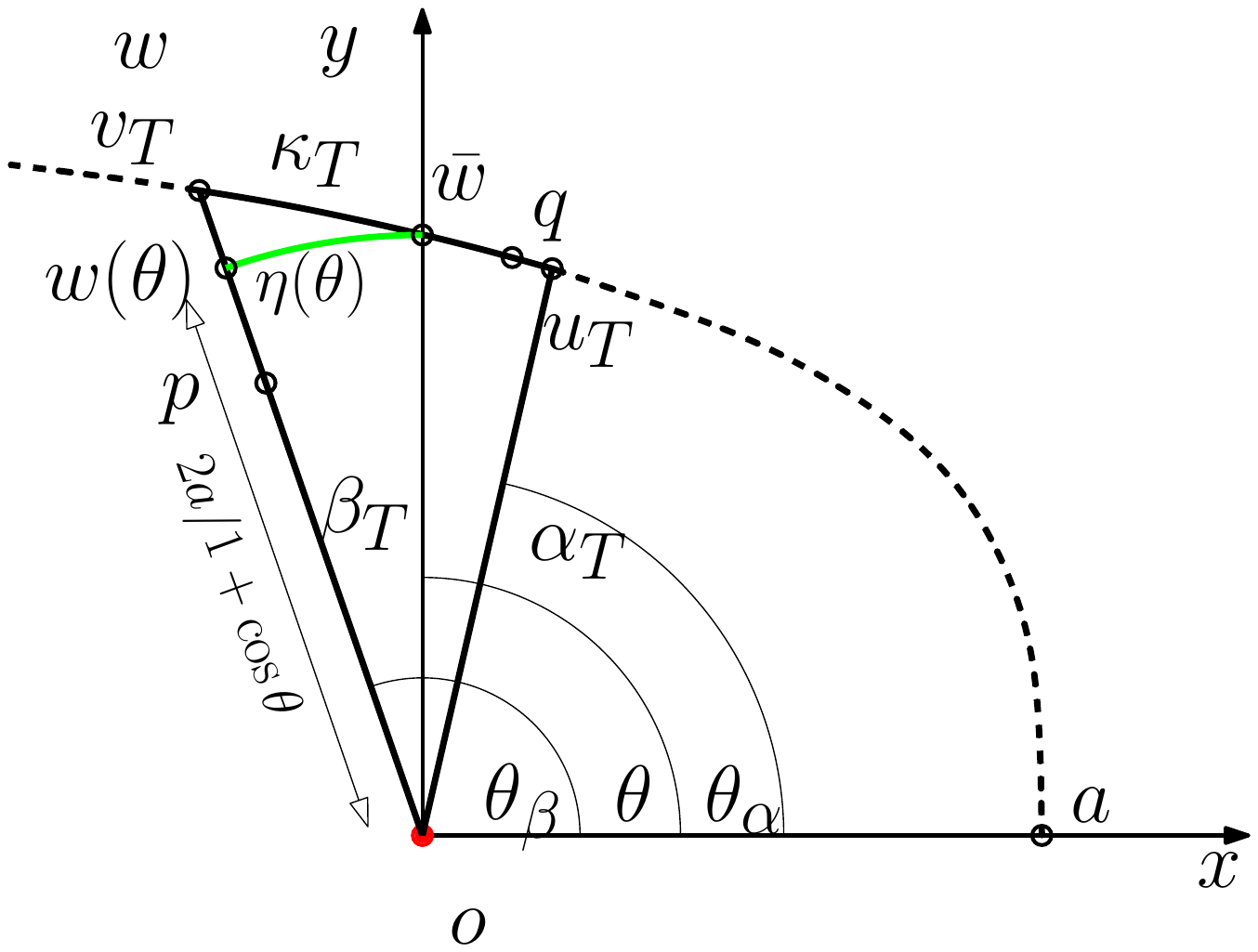}
          \label{fig:good_anchors_1bii}
          }
          \caption{\sf Different cases considered in the proof of
          Lemma~\ref{lem:good_anchors}, Case 1: $o$ is a vertex.}
          \label{fig:good_anchors_1}
        \end{center}
      \end{figure*}
      \begin{enumerate}[wide,itemindent=0pt]
        \item[Case 1(a): $\bar{w}^* \in \alpha_T$.]
          We parameterize the anchor point~$w$ by its clearance, i.e. $\cl(w(t)) = t$, and the
          feasible range of~$t$ is $[\cl(p), \cl(u_T)]$.
          If $\bar{w}(t) \in \alpha_T$, then the cost of $\eta(t)$ is simply $\theta_\beta -
          \theta_\alpha$, and
          $$\mu(t) = \ln \frac{\cl(w(t))}{\cl(p)} + \theta_\beta - \theta_\alpha = \ln
          \frac{t}{\cl(p)} + \theta_\beta - \theta_\alpha.$$
          Therefore, $\mu$ is minimized in the range $[\cl(p), \cl(u_T)]$ for $t = \cl(p)$, so $w^*
          = p$ in this case.
        \item[Case 1(b): $\bar{w}^* \in \kappa_T$.]
          We parameterize $w$ with the angle~$\theta$ of the segment~$o\bar{w}$.
          We call~$\theta$ \emph{feasible} if $\cl(w(\theta)) \geq \cl(p)$ and $\theta_\alpha \leq
          \theta \leq \theta_\beta$.
          We divide this case further into two subcases:
          \begin{enumerate}[wide,itemindent=0pt]
            \item[Case 1(b)(i): $\kappa_T$ is a line segment.]
              \WLOG, $\kappa_T$ is supported by the line $x = a$.
              The equation of the line in polar coordinates is $r = a/\cos \theta$.
              We have $\theta_\beta \leq \pi/2$.
              Restricting ourselves to feasible values of~$\theta$, we have
              $$\mu(\theta)
              =\ln \frac{\cl(w(\theta))}{\cl(p)} + \theta_\beta - \theta
              =\ln \frac{a / \cos \theta}{\cl(p)}  + \theta_\beta - \theta.$$
              Taking the derivative, we obtain
              $$\diff{}{\theta} \mu(\theta) = \tan \theta -1.$$
              This expression is negative for~$\theta = 0$, positive near~$\theta = \pi/2$, and it
              has at most one root within feasible values of~$\theta$, namely at~$\theta = \pi/4$.
              Therefore,~$\mu(\theta)$ is minimized when either~$\cl(w(\theta)) = \cl(p)$
              or~$\theta = \theta^* = \min\{\max\{\pi/4, \theta_\alpha\},\theta_\beta\}$.
              We pick $w_{\kappa_T} = w(\theta^*)$.
            \item[Case 1(b)(ii): $\kappa_T$ is a parabolic arc.]
              \WLOG, the parabola supporting~$\kappa_T$ is equidistant between~$o$ and the line~$x
              = 2a$.
              The equation of the parabola in polar coordinates is~$r = 2a/(1 + \cos \theta)$.
              We have~$\theta_\beta \leq \pi$.
              The polar coordiantes of $w(\theta)$ are $w(\theta) = (\frac{2a}{1+\cos \theta},
              \theta_\beta)$.

              Restricting ourselves to feasible values of~$\theta$, we have
              $$ \mu(\theta)
              =\ln \frac{\cl(w(\theta))}{\cl(p)} + \theta_\beta - \theta
              =\ln \frac{2a/(1+\cos \theta)}{\cl(p)} + \theta_\beta - \theta.$$
              Here,
              $$ \diff{}{\theta} \mu(\theta)
              =\frac{\sin \theta}{1 + \cos \theta} - 1
              =\tan(\theta/2) -1.$$
              Again, the expression is negative for~$\theta = 0$, positive for~$\theta$
              near~$\pi$, and it has at most one root within feasible values of~$\theta$, namely
              at~$\theta = \pi/2$.
              Therefore,~$\mu(\theta)$ is minimized when either~$\cl(w(\theta)) = \cl(p)$ or
              $\theta = \theta^* = \min\{\max\{\pi/2,\theta_{{\alpha_T}}\},\theta_{\beta_T}\}$.
              We pick~$w_{\kappa_T} = w(\theta^*)$.
          \end{enumerate}
      \end{enumerate}
    \item[Case 2: $o$ is an edge.]
      \WLOG, $o$ lies on the line~$y = 0$, the edge~${\alpha_T}$ lies on the line~$x = x_\alpha$,
      the edge~${\beta_T}$ lies on the line~$x = x_\beta$, and~$x_\beta > x_\alpha \geq 0$.
      In this case, $\eta_w$ is a horizontal segment.
      We again consider two cases.
      See Figure~\ref{fig:good_anchors_2}.
      \begin{figure*}[t]
        \begin{center}
          \subfloat
          [\sf Case 2(a): $\bar{w}^* \in \alpha_T$]
          { 
          \includegraphics[scale=.34]{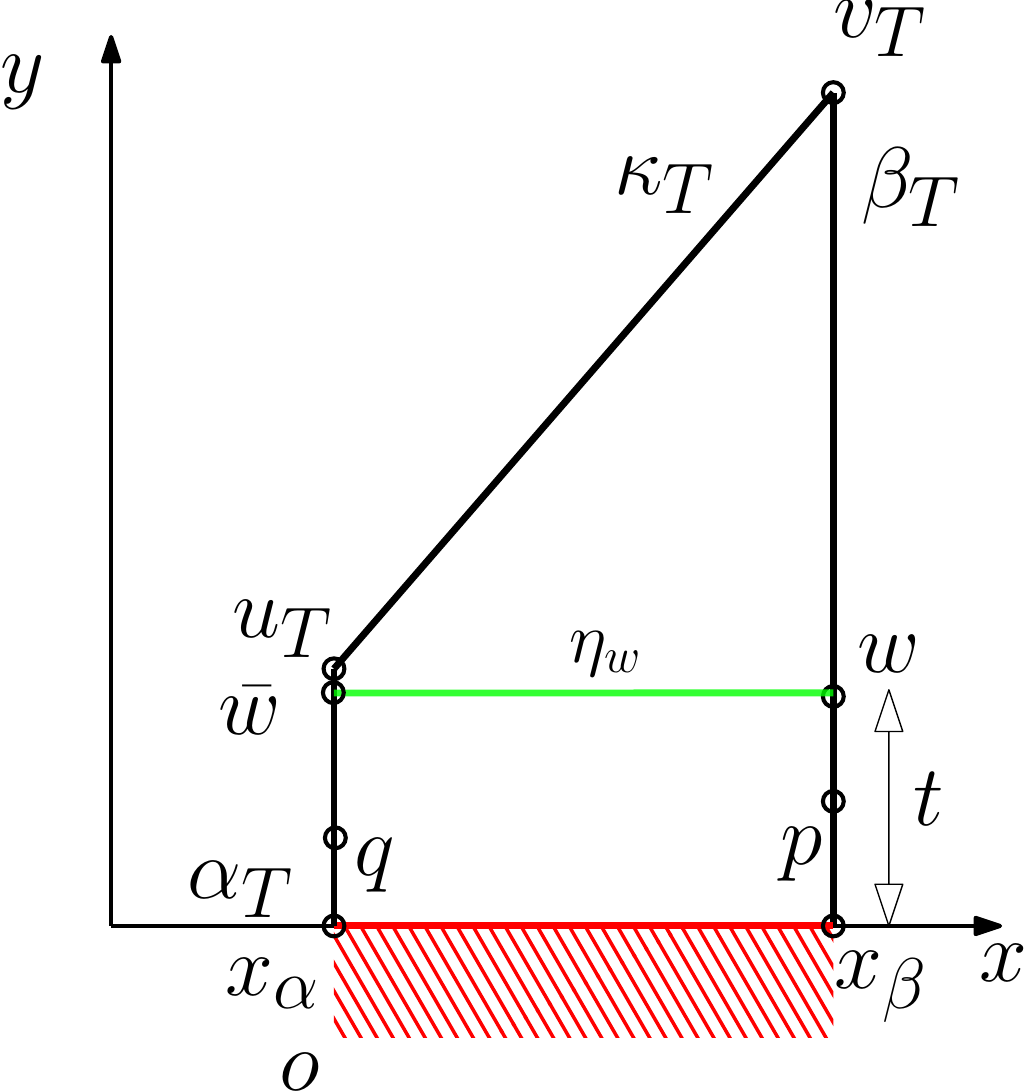}
          \label{fig:good_anchors_1a}
          }
          \hspace{0.03\textwidth}
          \subfloat
          [\sf Case 2(b)(i): $\bar{w}^* \in \kappa_T$; $\kappa_T$ is a line segment]
          {
          \includegraphics[scale=.34]{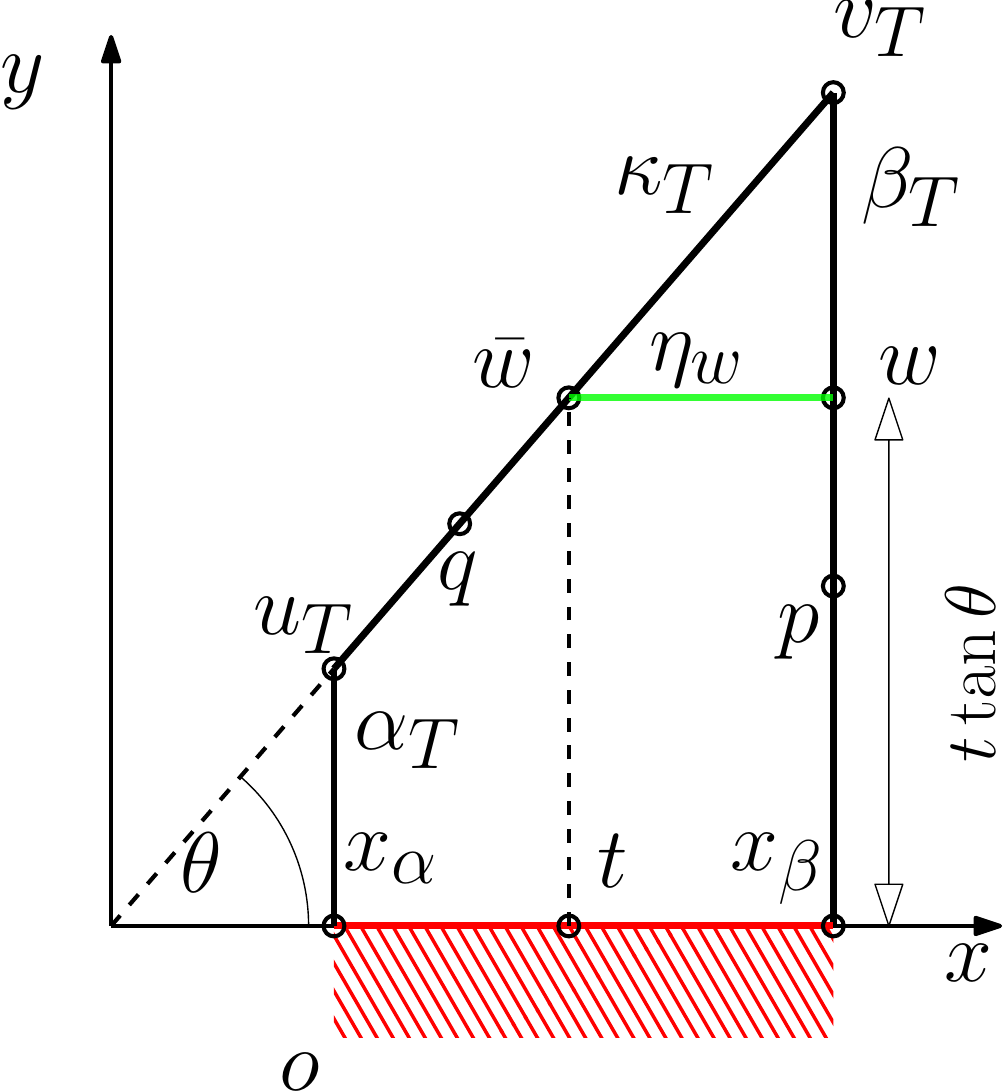}
          \label{fig:good_anchors_1bi}
          }
          \hspace{0.01\textwidth}
          \subfloat
          [\sf Case 2(b)(ii): $\bar{w}^* \in \kappa_T$; $\kappa_T$ is a parabolic arc]
          {
          \includegraphics[scale=.34]{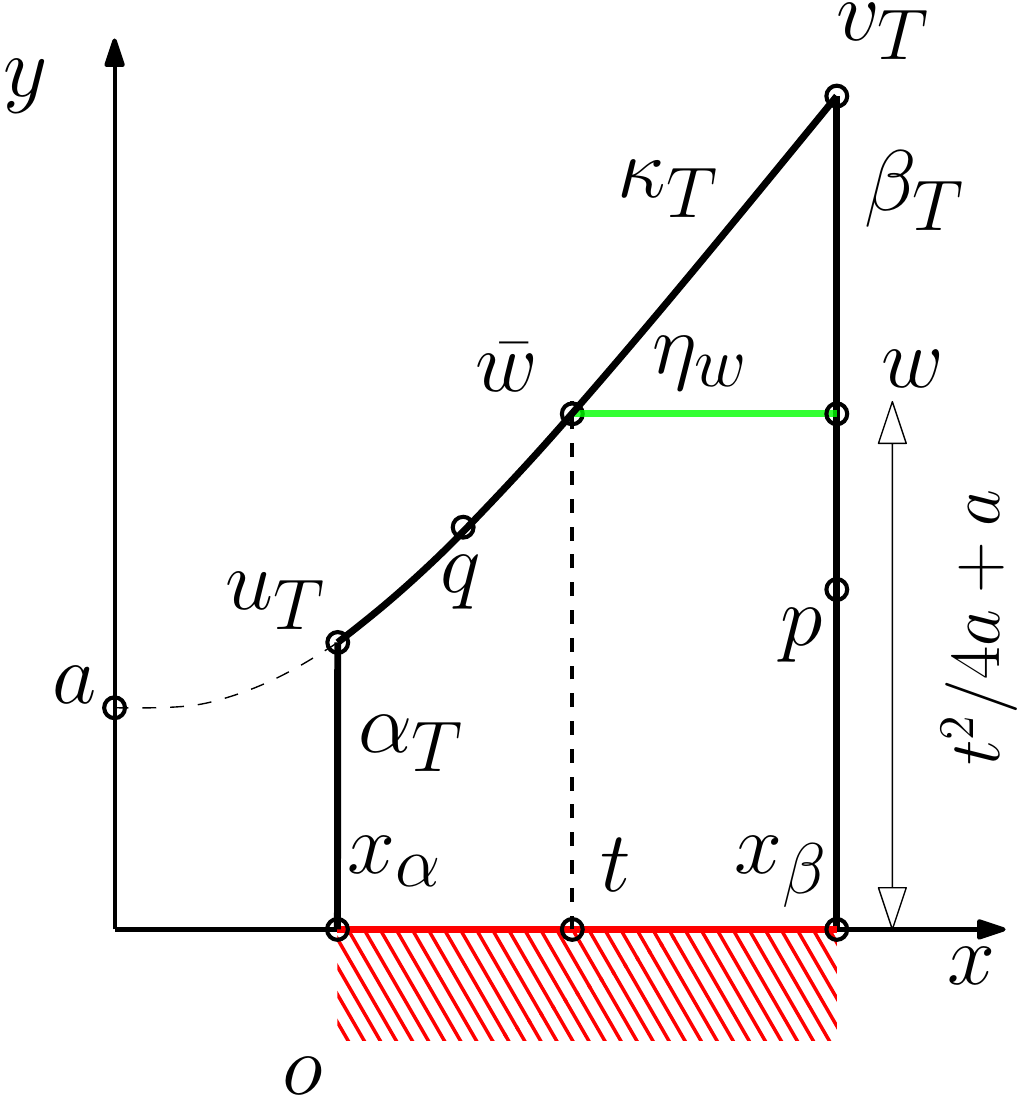}
          \label{fig:good_anchors_1bii}
          }
          \caption{\sf Different cases considered in the proof of
          Lemma~\ref{lem:good_anchors}, Case 2: $o$ is an edge.}
          \label{fig:good_anchors_2}
        \end{center}
      \end{figure*}
      \begin{enumerate}[wide,itemindent=0pt]
        \item[Case 2(a): $\bar{w}^* \in \alpha_T$.]
          As in Case 1(a), we parameterize the anchor point~$w$ by its clearance, i.e., $\cl(w(t)) =
          t$, and the feasible range of~$t$ is $[\cl(p), \cl(u_T)]$.
          Restricting ourselves to feasible values of~$w(t)$, we have
          $$ \mu(\gamma(p,\lambda)) 
          =\ln \frac{\cl(w(t))}{\cl(p)} + \frac{\dist{\eta(t)}}{\cl{(w(t))}}
          =\ln \frac{t}{\cl(p)} + \frac{x_\beta - x_\alpha}{t}.$$
          We have
          $$\diff{}{t} \mu(t) = 1/t - \frac{(x_\beta - x_\alpha)}{t^2}.$$
          This expression is negative for $t$ near~$0$, positive for large~$t$, and it has at most
          one root within feasible values of~$t$, namely at~$t = x_\beta - x_\alpha$.
          If $x_\beta - x_\alpha \leq \cl(u_T)$, then~$\mu(t)$ is minimized when either~$\cl(w(t)) =
          \cl(p)$ or~$t = t^* = x_\beta - x_\alpha$.
          If $x_\beta - x_\alpha > \cl(u_T)$, then $\bar{w}^* \in \kappa_T$, so assume that $x_\beta
          - x_\alpha \leq \cl(u_T)$.
          We pick~$w_{\alpha_T} = w(t^*)$.
      \end{enumerate}
      \begin{enumerate}[wide,itemindent=0pt]
        \item[Case 2(b): $\bar{w}^* \in \kappa_T$.]
          We parameterize~$w$ by the $x$-coordinate of $\bar{w}$.
          We call $t$ feasible if $\cl(w(t)) \geq \cl(p)$ and $t \in [x_\alpha, x_\beta]$.
          There are two subcases.
          \begin{enumerate}[wide,itemindent=0pt]
            \item[Case 2(b)(i): $\kappa_T$ is a line segment.]
              \WLOG, the line supporting~$\kappa_T$ intersects~$o$ at the origin with angle~$\theta$.
              Restricting ourselves to feasible values of~$t$, we have
              $$ \mu(t) 
              =\ln \frac{\cl(w(t))}{\cl(p)} + \frac{\dist{\eta(t)}}{\cl{(w(t))}}
              =\ln \frac{t \tan \theta}{\cl(p)} + \frac{x_\beta - t}{t \tan \theta}.$$
              We see
              $$\diff{}{t} \mu(t) = \frac{1}{t} - \frac{x_\beta}{t^2 \tan \theta}
              = \frac{ \tan \theta - x_\beta}{t^2 \tan \theta}.$$
              This expression is negative for $t$ near~$0$, positive for large~$t$, and it has at most
              one root within feasible values of~$t$, namely at~$t = x_\beta/\tan \theta$.
              Therefore,~$\mu(t)$ is minimized when either~$\cl(w(t)) = \cl(p)$ or~$t = t^* =
              \min\{\max\{x_\beta/\tan \theta, x_\alpha\},x_\beta\}$.
              We pick~$w_{\kappa_T} = w(t^*)$;
              note that $\cl(w_{\kappa_T}) = t^* \tan \theta$.
            \item[Case 2(b)(ii): $\kappa_T$ is a parabolic arc.]
              \WLOG, the parabola supporting~$\kappa_T$ is equidistant between~$o$ and a point located
              at~$(0, 2a)$.
              Therefore, the parabola is described by the equation~$y = x^2/(4a) + a$.

              Restricting ourselves to feasible values of~$t$, we have
              $$ \mu(t) 
              =\ln \frac{\cl(w(t))}{\cl(p)} + \frac{\dist{\eta(t)}}{\cl{(w(t))}}
              =\ln \frac{t^2/(4a) + a}{\cl(p)} + \frac{x_\beta - t}{t^2/(4a) + a}.$$
              We have
              $$\diff{}{t} \mu(t) = \frac{2t^3 + 4at^2 + 8a(a - x_\beta)t - 16a^3}
              {(t^2 + 4a^2)^2}.$$
              This expression is negative for~$t$ near~$0$ and positive for large~$t$.
              The derivative of the numerator is~$6t^2 + 8at + 8a(a - x_\beta)$, which has at most
              one positive root.
              Therefore, the numerator has at most one positive local maximum or minimum.
              We see $\diff{}{t} \mu(t)$ goes from negative to positive around exactly one positive
              root (which may not be feasible), and~$\mu(t)$ has one minimum at a positive value
              of~$t$.
              Let~$t'$ be this root of~$\diff{}{t} \mu(t)$.
              Value~$\mu(t)$ is minimized when either~$\cl(w(t)) = \cl(p)$ or $t = t^* =
              \min\{\max\{t', x_\alpha\},x_\beta\}$.
              We pick~$w_{\kappa_T} = w(t^*)$;
              note that $\cl(w_{\kappa_T}) = (t^*)^2/(4a) + a$.
          \end{enumerate}
      \end{enumerate}
  \end{enumerate}

  We note that in all cases $w^* \in \set{p, w_{\kappa_T}, w_{\alpha_T}}$ for some choices of
  $w_{\kappa_T}$ and $w_{\alpha_T}$ that depend only on the geometry of~$T$ and not on~$p$.
  Note that no subcase of Case 1 required picking a concrete $w_{\alpha_T}$, so if $o$ is a
  vertex, we let $w_{\alpha_T}$ be an arbitrary point on $\beta_T$.
  In every case, $w_{\kappa_T}$ and $w_{\alpha_T}$ can be computed in~$O(1)$ time.
  We conclude the proof of the lemma.
\end{proof}

\section{Approximation Algorithms}
\label{sec:ptas}

In this section, we propose a near-quadratic-time $(1+\eps)$-approximation algorithm for
computing the minimal-cost path between two points $s,t \in \calF$ amid $\calO$.
We assume that $\cl(s) \leq \cl(t)$ throughout this section.
We first give a high-level overview of the algorithm and then describe each step in detail.
Throughout this section, let~$\Pi^*$ denote a minimal-cost $(s,t)$-path.

\noindent
\paragraph{High-level description.}
Our algorithm begins by computing the refined Voronoi diagram~$\dvd$ of~$\calO$.
The  algorithm then works in three stages.
The first stage computes an $O(n)$-approximation of~$\OPT = \mu (s,t)$,
\ie, it returns a value~$\tilde{d}$ such that~$\OPT \leq \tilde{d} \leq cn\OPT$ for some constant~$c >0$.
By augmenting~$\dvd$ with a linear number of additional edges, each a constant-clearance path
between two points on the boundary of a cell of~$\dvd$, the algorithm constructs a graph~$G_1$
with~$O(n)$ vertices and edges, and it computes a minimal-cost path from~$s$ to~$t$ in~$G_1$. 

Equipped with the value~$\tilde{d}$, the second stage computes an $O(1)$-approximation of~$\OPT$.
For a given~$d \geq 0$, this algorithm constructs a graph~$G_2 = G_2[d]$ by sampling~$O(n)$ points
on the boundary of each cell~$T$ of~$\dvd$ and connecting these sample points by adding~$O(n)$
edges (besides the boundary of~$T$), each of which is again a constant-clearance path.
The resulting graph~$G_2$ is planar and has~$O(n^2)$ edges total, so a minimal-cost path in~$G_2$
from~$s$ to~$t$ can be computed in~$O(n^2)$ time~\cite{HKRS97}.
We show that if $d \geq \OPT$, then the cost of the optimal path from~$s$ to~$t$ in~$G_2$
is~$O(d)$.
Therefore, if $d \in [\OPT, 2\OPT]$, the cost of the optimal path is~$O(\OPT)$.
Using the value of~$\tilde{d}$, we run the above procedure for~$O(\log n)$ different values
of~$d$, namely $d \in \{\tilde{d}/2^i \mid 0 \leq i \leq \lceil \log_2 cn \rceil\}$, and return
the least costly path among them.
Let~$\hat{d}$ be the cost of the path returned.

Finally, using the value~$\hat{d}$, the third stage samples~$O(n/\eps)$ points on the boundary of
each cell~$T$ of~$\dvd$ and connects each point to $O((1/\eps) \log (n/\eps))$ other points on the
boundary of~$T$ by an edge.
Unlike the last two stages, each edge is no longer a constant-clearance path but it is a
minimal-cost path between its endpoints lying inside~$T$.
The resulting graph~$G_3$ has~$O(n^2 / \eps)$ vertices and~$O((n^2 / \eps^2) \log (n / \eps))$
edges.
The overall algorithm returns the minimal-cost path in~$G_3$.
Anchor points and well-behaved paths play a pivotal role in each stage of the algorithm.

\subsection{Computing an $O(n)$-approximation algorithm}
\label{subsec:O(n)-apx}
Here, we describe a near-linear time algorithm to obtain an $O(n)$-approximation of~$\OPT$.
We augment~$\dvd$ with~$O(n)$ additional edges as described below to create the
graph~$G_1 = (V_1, E_1)$.

We do the following for each cell~$T$ of~$\dvd$.
We compute anchor points~$w_{\alpha_T}$ and~$w_{\kappa_T}$ as described in
Lemma~\ref{lem:good_anchors}.
Let~$s_T$ be the point on~$\beta_T$ of clearance~$\min\{\cl(v_T), \cl(s)\}$.
Set $W_T = \set{w_{\alpha_T}, w_{\kappa_T}, s_T}$, $E_T = \set{\eta_w | w \in W_T}$, and 
$\bar{W}_T = \set{\bar{w}_{\alpha_T}, \bar{w}_{\kappa_T}, \bar{s}_T}$.
Vertex set~$V_1$ is the set of Voronoi vertices plus the set $W_T \cup \bar{W}_T$ for all cells $T
\in \dvd$%
\footnote{Note that as we consider each cell independently, we actually consider each edge $e$
twice as it is adjacent to two cells and add vertices on $e$ for each cell independently.
The set of vertices put on $e$ is the union of these two sets.
Considering each edge twice does not change the complexity of the algorithm or its analysis, and
doing so simplifies the algorithm's description.}.
Next, for each edge $e$ of $\dvd$, we add the portion of $e$ between two consecutive vertices of
$V_1$ as an edge of $E_1$, and for each cell $T \in \dvd$ we also add $E_T$ to $E_1$.
See Figure~\ref{fig:n-approx}(a) and (b).
(Note that if $s_T = v_T$ then $\eta_{s_T}$ is a trivial path and there is no need to add
$\eta_{s_T}$ to $E_1$.
Paths $\eta_{w_{\alpha_T}}$ and $\eta_{w_{\kappa_T}}$ may be trivial as well.)
The cost $\mu(e)$ for each edge $e \in E_1$ is computed using~\eqref{eq:length} or the equations
of Wein \etal~\cite{WBH08} for Voronoi edges.
By construction, $|V_1| = O(n)$ and $|E_1| = O(n)$.
We compute and return, in $O(n \log n)$ time, an optimal path from $s$ to $t$ in $G_1$.

\begin{figure*}[tb]
  \centering
  \subfloat
   [\sf Point-obstacle cell]
   { 
    \hspace{.05\textwidth}
   	\includegraphics[scale=.34]{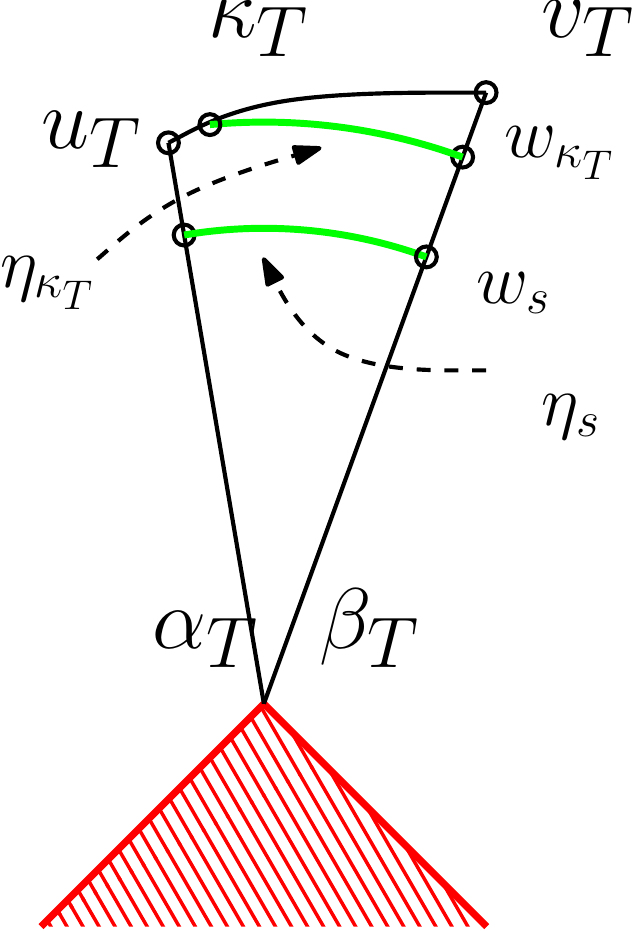}
    \hspace{.05\textwidth}
   }
  \subfloat
   [\sf Line-obstacle cell]
   {
   	\includegraphics[scale=.34]{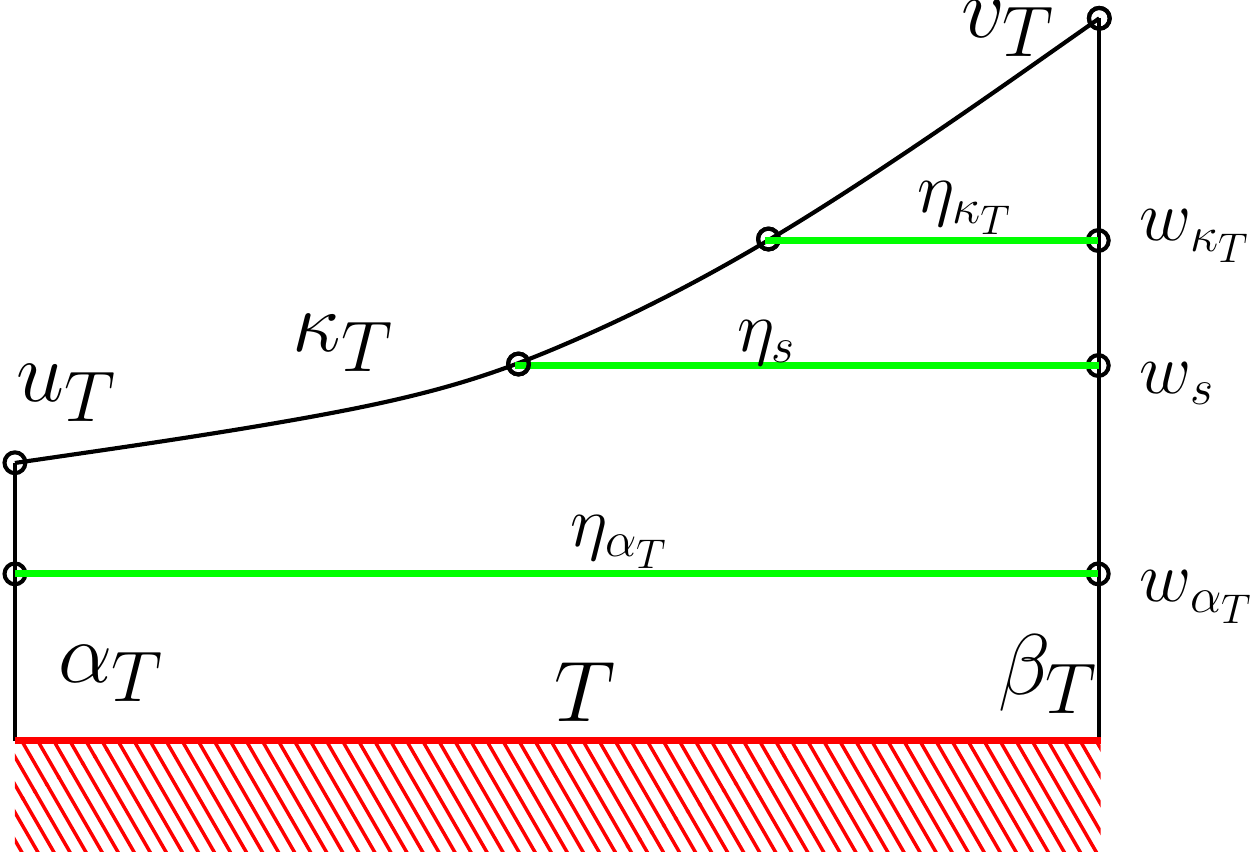}
   }
  \subfloat
   [\sf $(p,q)$-path $\gamma$]
   {
     \includegraphics[scale=.34]{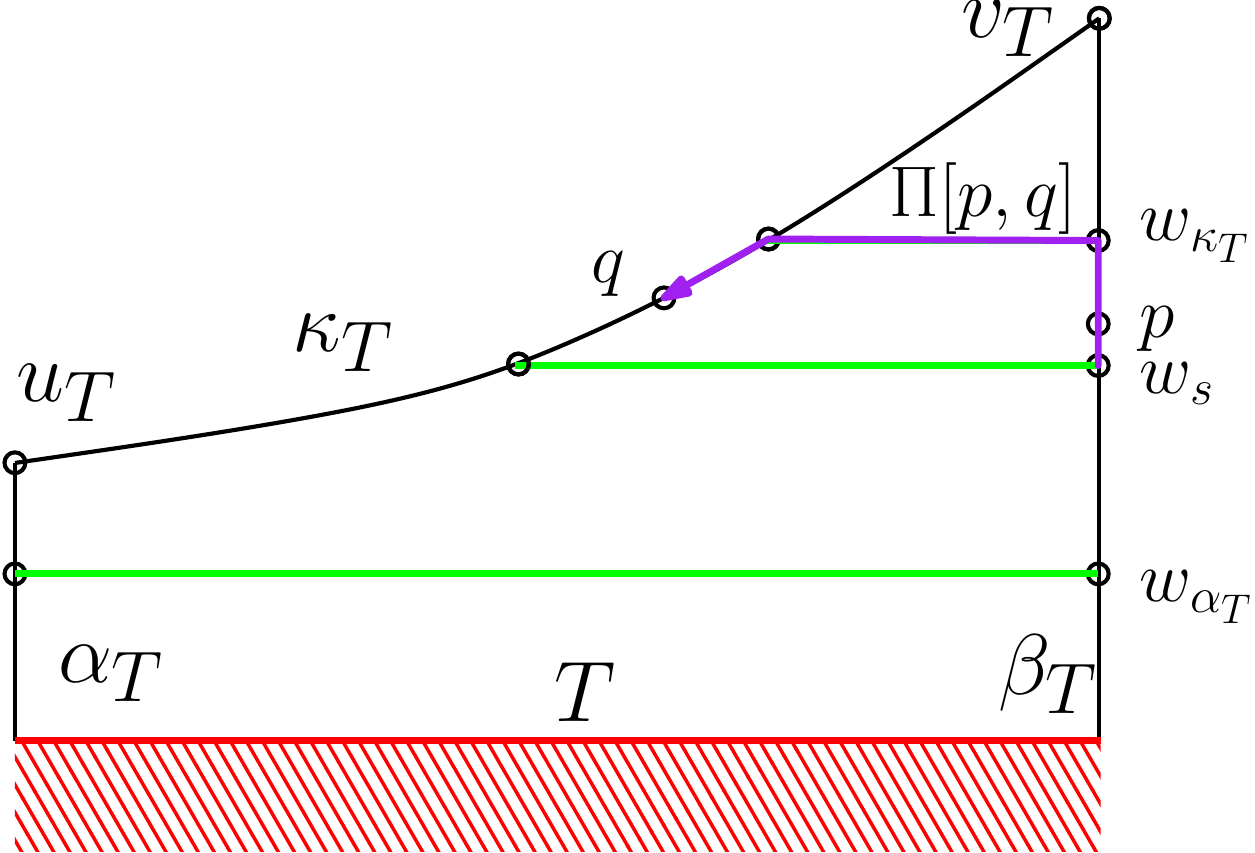}
   }
  \caption{\sf (a), (b) Edges added within cell~$T$ for the~$O(n)$-approximation algorithm.
  (c) Visualization of proof of Lemma~\ref{lem:n-approximate}.}
  \label{fig:n-approx}
\end{figure*}

\begin{lemma}
\label{lem:n-approximate}
Graph~$G_1$ contains an $s,t$-path of cost at most~$O(n) \cdot \OPT$.
\end{lemma}
\begin{proof}
  Let~$\Pi^*$ be an optimal path from $s$ to $t$.
  We will deform~$\Pi^*$ into another path~$\tilde{\Pi}$ from $s$ to $t$ of cost $O(n) \cdot
  \OPT$ that enters or exits the interior of a cell of $\dvd$ only at the vertices of $V_1$ and
  follows an arc of $E_T$ in the interior of the cell $T$.
  By construction, $\tilde{\Pi}$ will be a path in $G_1$ which will imply the claim.

  By construction $s,t \in V_1$.
  Let $\Pi$ denote the current path that we have obtained by deforming $\Pi^*$.
  Let $T \in \dvd$ be the first cell (along $\Pi$) such that $\Pi$ enters the interior of $T$ but
  $\text{int}(T) \cap \Pi$ is not an arc of $E_T$.
  Let $p$ (resp. $q$) be the first (resp. last) point of $\Pi \cap T$.
  If both $p$ and $q$ lie on the same edge of $T$ or neither of them lies on $\beta_T$, we replace
  $\Pi[p,q]$ with the well-behaved path $\gamma(p,q)$, because $\gamma(p,q) \subseteq \partial T$
  in this case.
  Suppose $p \in \beta_T$ and $q \notin \beta_T$ (the other case is symmetric).
  We replace $\Pi[p,q]$ with $\Pi[p,q] = p s_T \circ \gamma(s_T, q)$, i.e., the segment $p s_T
  \subseteq \beta_T$ followed by the well-behaved path $\gamma(s_T, q)$.
  By Lemma~\ref{lem:good_anchors}, $\clo(\text{int}(T) \cap \gamma(s_t,q)) \in E_T$.

  We repeate the above step until no such cell $T$ is left.
  Since the above step is performed at most once for each cell of $\dvd$, we obtain the final path
  $\tilde{\Pi}$ in $O(n)$ steps.

  We now bound the cost of $\tilde{\Pi}$.
  If $\Pi[p,q]$ is replaced by $\gamma(p,q)$, then by Lemma~\ref{lem:easy_well-behaved_cost},
  $\mu(\gamma(p,q)) \leq 3 \mu(\Pi[p,q])$.
  On the other hand, if $p \in \beta_T$ and $q \notin \beta_T$, then $\mu(\tilde{\Pi}[p,q]) =
  \mu(p s_T) + \mu(\gamma(s_T, q))$.
  By the triangle inequality, $\pi(s_T, q) \leq \pi(s_T, p) + \pi(p,q)$, and by
  Lemma~\ref{lem:hard_well-behaved_cost}, 
  \[\mu(\gamma(s_T, q)) \leq 11 \pi(s_T, q) \leq 11(\pi(s_T, p) + \pi(p,q)).\]
  By Corollary~\ref{cor:clearance_bound}, $\pi(p, s_T) = \mu(p s_T) \leq \OPT$.
  Putting everything together,
  \begin{align}
    \mu(\tilde{\Pi}[p,q]) &\leq 12\pi(p,s_T) + 11\pi(p,q)\\
    &\leq 12\OPT + 11\pi(p,q) = O(\OPT). \nonumber
  \end{align}
  Summing over all $O(n)$ steps, the cost of $\tilde{\Pi}$ is $O(n) \cdot \OPT$.
\end{proof}

We thus obtain the following.
\begin{theorem}
\label{thm:easy_approx}
	Let~$\calO$ be a set of polygonal obstacles in the plane, and 
	let~$s,t$ be two points outside~$\calO$.
	There exists an $O(n \log n)$-time 
	$O(n)$-approximation algorithm for computing the minimal-cost 
	path between~$s$ and~$t$.
%
\end{theorem}

\subsection{Computing a constant-factor approximation}
\label{subsec:constant}
Recall that, given an estimate~$d$ of the cost~$\OPT$ of the optimal path, we construct a planar
graph~$G_2 = G_2[d]$ by sampling points along the edges of the refined Voronoi diagram~\dvd.
The sampling procedure here can be thought of as a warm-up for the more-involved sampling
procedure given in Section~\ref{subsec:ptas}.

\begin{figure}[tb]
\centering
  \centering
   	\includegraphics[scale=.34]{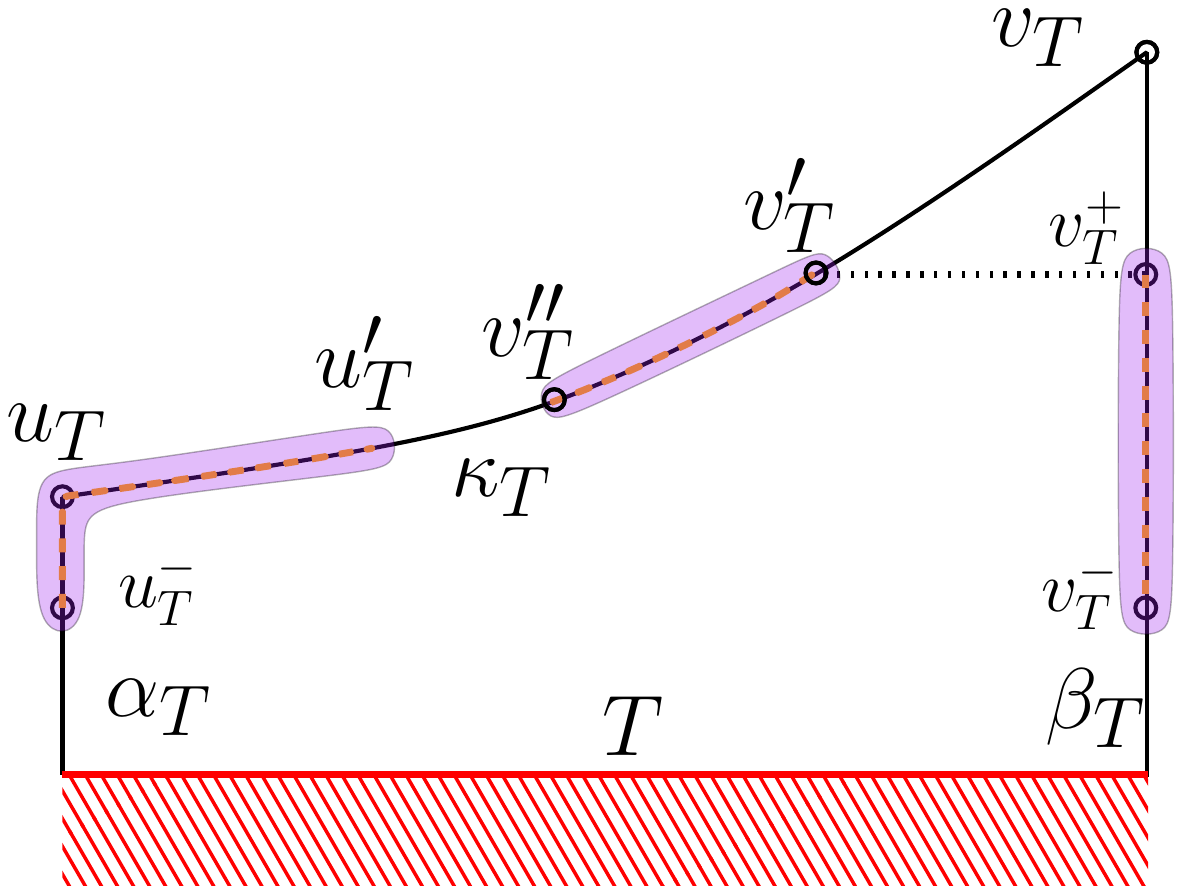}
	\caption{\sf
				Samples placed on the edges of a cell~$T$ of~$\dvd$. 
				The sampled regions are depicted in purple.
        For the constant-factor approximation algorithm, samples are placed on~$\beta_T$ only.}
  \label{fig:samples}
\end{figure}
Let~$T$ be a Voronoi cell of~$\dvd$.
Let~$v_T^-$ and~$v_T^+$ be the points on~$\beta_T$ with clearance $\min
\{\cl(v_T), \cl(t)/\exp(d)\}$ and $\min \{\cl(v_T), \cl(s) \cdot \exp(d)\}$, respectively.
We refer to the segment $\hat{\beta}_T = v_T^- v_T^+ \subseteq \beta_T$ as the \emph{marked
portion} of $\beta_T$.
By \eqref{eq:length_away}, $\mu(\hat{\beta}_T) = O(d)$.
We place sample points on $\hat{\beta}_T$, its endpoints always being sampled, so that the cost
between consecutive samples is exactly~$\frac{d}{n}$ (except possibly at one endpoint).
Given a sample point~$p$ on an edge of~$\dvd$, it is straightforward to compute the coordinates of
the sample point $p'$ on the same edge such that $\pi(p,p') = c$ for any~$c > 0$.
Simply use the formula for the cost along a Voronoi edge given in~\cite[Corollary 8]{WBH08}.
We emphasize that the points are separated evenly by \emph{cost}; the samples are not uniformly
placed by Euclidean distance along the edge; see Figure~\ref{fig:samples}.

For each cell $T \in \dvd$, let $W_T$ be the set of sample points on $\beta_T$ plus the anchor
points $w_{\kappa_T}$ and $w_{\alpha_T}$.
For each point $w \in W_T$, we compute the constant-clearance arc $\eta_w$.
Let $E_T = \set{\eta_w | w \in W_T}$ and $\bar{W}_T = \set{\bar{w} | w \in W_T}$ be the set of
other endpoints of arcs in $E_T$.
Set $V_2$ is the set of vertices of $\dvd$ plus the set $W_T \cup \bar{W}_T$ over all cells in
$\dvd$.
For each edge of $\dvd$, we add the portions between consecutive sample vertices of $V_2$ to
$E_2$, and we also add $E_T$, over all cells $T \in \dvd$, to $E_2$.
The cost of each edge in $E_2$ is computed as before.
We have $|V_2|,|E_2| = O(n^2)$, and $G_2$ can be constructed in $O(n^2)$ time.

The refined Voronoi diagram~$\dvd$ is planar.
Every edge~$\eta_w$ added to create~$G_2$ stays within a single cell of~$\dvd$ and has constant
clearance.
Therefore, no new crossings are created during its construction, and~$G_2$ is planar as well.  We
compute the minimal-cost path from~$s$ to~$t$ in~$G_2$, in $O(n^2)$ time, using the algorithm of
Henzinger \etal~\cite{HKRS97}.

\begin{lemma}
\label{lem:sample_bounds}
If $d \geq \OPT$, then $\Pi^* \cap \beta_T \subseteq \hat{\beta}_T$ for any cell $T \in \dvd$.
\end{lemma}
\begin{proof}
Let $p_{\text{min}}$ be the point where~$\Pi^*$ attains the minimal clearance.
Clearly, $\pi(s,t) \geq \pi(s,p_{\text{min}}) + \pi(p_{\text{min}}, t)$. 
Using this observation together with 
Lemma~\ref{lem:cost_observations}(i) and the assumption that $\cl(s) \leq \cl(t)$, we conclude
that the clearance of any point on $\Pi^*$ is at least $\cl(t) / \exp(d^*) \geq \cl(t) / \exp(d)$.
A similar argument implies the clearance of any point on $\Pi^*$ is at most $\cl(s) \cdot \exp(d)$.
Hence, $\Pi^* \cap \beta_T \subseteq \hat{\beta}_T$.
\end{proof}

\begin{lemma}
\label{lem:constant-approximate}
For~$d \geq d^*$, graph~$G_2[d]$ contains an $s,t$-path of cost~$O(d)$.
\end{lemma}
\begin{proof}
  We deform the optimal path~$\Pi^*$ into a path $\tilde{\Pi}$ of~$G_2$ in the same way as in the
  proof of Lemma~\ref{lem:n-approximate} except for the following twist.
  As in Lemma~\ref{lem:n-approximate}, let $p$ (resp. $q$) be the first (resp. last) point on
  $\Pi^*$ in a cell $T$ of $\dvd$.
  If $p \in \beta_T$ and $q \notin \beta_T$, let $p'$ be a sample point on $\beta_T$ such that
  $\pi(p,p') \leq d/n$; the existence of $p'$ follows from Lemma~\ref{lem:sample_bounds}.
  We replace $\Pi^*[p,q]$ with $\tilde{\Pi}_T = p p' \circ \gamma(p',q)$, i.e., $p'$ replaces the role of
  $s_T$ in the proof of Lemma~\ref{lem:n-approximate}.
  Since $\pi(p,p') \leq d/n$, we have
  $$\mu(\tilde{\Pi}_T) \leq 11\pi(p,q) + O(d/n).$$
  Summing over all steps in the deformation of $\Pi^*$ and using the fact $d \geq \OPT =
  \mu(\Pi^*)$, we obtain $\mu(\tilde{\Pi}) = O(d)$.
  It is clear from the construction that $\tilde{\Pi}$ is a path in $G_2$.
\end{proof}


For our constant-factor approximation algorithm, we perform an exponential search over the values
of path costs.
Let~$\tilde{d} \leq cn\OPT$ be the cost of the path returned by the $O(n)$-approximation algorithm (Section~\ref{subsec:O(n)-apx}).
For each~$i$ from~$0$ to~$\lceil\log cn\rceil$, we choose $d_i = \tilde{d} / 2^i$.
We run the above procedure to construct a graph~$G_2[d_i]$ and compute a minimal-cost path~$\Pi_i$
in the graph.
Let $\Delta_i = \mu(\Pi_i)$.
We compute $k = \argmin_i \Delta_i$ and return $\Pi_k$.

Fix integer~$\hat{i}$ so~$\OPT \leq d_{\hat{i}} \leq 2\OPT$.
By Lemma~\ref{lem:constant-approximate}, we have
$$ \Delta_k \leq \Delta_{\hat{i}} = O(d_{\hat{i}}) = O(\OPT).$$

%
%
%

\begin{theorem}
\label{thm:constant}
	Let~$\calO$ be a set of polygonal obstacles in the plane, and let~$s,t$ be two points outside~$\calO$.
There exists an~$O(n^2\log n)$ time $O(1)$-approximation algorithm for computing the minimal-cost path between~$s$ and~$t$.
\end{theorem}

\subsection{Computing the final approximation}
\label{subsec:ptas}
Finally, let~$d$ be the estimate returned by our constant factor approximation algorithm so
that~$\OPT \leq d \leq c\OPT$ for some constant~$c$.
We construct a graph~$G_3 = (V_3, E_3)$ by sampling $O(n/\eps)$ points along each edge of $\dvd$
and connecting (a certain choice of) $O(\frac{n}{\eps^2} \log \frac{n}{\eps})$ pairs of sample
points on the boundary of each cell of $\dvd$ by ``locally optimal'' paths.
We guarantee $|V_3| = O(\frac{n^2}{\eps})$ and $|E_3| = O(\frac{n^2}{\eps^2}\log \frac{n}{\eps})$.
We compute and return, in $O(\frac{n^2}{\eps^2} \log \frac{n}{\eps})$ time, a minimal-cost path in
$G_3$~\cite{FT87}.

\noindent
\paragraph{Vertices of $G_3$.}

Let $\underline{c} = \cl(t) / \exp(d)$ and $\overline{c} = \cl(s) \cdot \exp(d)$.
Let $T$ be a cell of $\dvd$.
For each edge of $T$, we mark at most two connected portions, each of cost $O(d)$.
We refer to each marked portion as an \emph{edgelet}.
We sample points on each edgelet so that two consecutive samples lie at cost $\eps d/n$ apart;
endpoints of each edgelet are always included in the sample.
The total number of samples places on $\partial T$ is $O(n/\eps)$.
We now describe the edgelets of $T$.

Let $u_T^-, u_T^+$ be points on $\alpha_T$ of clearance $\min \set{\cl(u_T), \underline{c}}$ and
$\min \set{\cl(u_T), \overline{c}}$, respectively.
Similarly, let $v_T^-,v_T^+$ be points on $\beta_T$ of clearance $\min \set{\cl(v_T),
\underline{c}}$ and $\min\set{\cl(v_T), \overline{c}}$, respectively.
The edgelets on $\alpha_T$ and $\beta_T$ are the segments $u_T^- u_T^+$ and $v_T^- v_T^+$,
respectively.
Next, we mark (at most) two edgelets on $\kappa_T$:
If $\mu(\kappa_T) \leq 2d$, then the entirety of $\kappa_T$ is a single edgelet;
otherwise, let $u'_T \in \kappa_T$ be the point such that $\mu(\kappa_T[u_T, u'_T]) = 2d$.
Let $v'_T \in \kappa_T$ be the point of clearance $\cl(v_T^+)$.
If $\mu(\kappa_T[u'_T,v'_T]) \leq 4d$, then $\kappa_T[u_T,v'_T]$ is the only edgelet on $\kappa_T$.
Otherwise, let $v''_T \in \kappa_T$ be the point such that $\cl(v''_T) \leq \cl(v'_T)$ and
$\mu(\kappa_T[v'_T, v''_T]) = 4d$; $\kappa_T$ has two edgelets $\kappa_T[u_T, u'_T]$ and
$\kappa_T[v'_T, v''_T]$.
See Figure~\ref{fig:samples}.
We repeat this procedure for all cells of $T$.
Set $V_3$ is the set of all samples placed on the edges of $\dvd$.
We have $|V_3| = O(n^2 / \eps)$.

\noindent
\paragraph{The edges of~$G_3$.}
Let~$T$ be a cell of~$\dvd$ incident to obstacle feature~$o$.
We say two points~$p$ and~$q$ in~$T$ are \emph{locally reachable} from one another if
the minimal-cost path from~$p$ to~$q$ relative only to~$o$ lies within~$T$.
Equivalently, the minimal-cost path relative to~$o$ is equal to the minimal-cost path relative
to~$\calO$.

Let~$p \in \partial T$ be a sample point.
We compute a subset $S(p) \subseteq V_3$ of candidate neighbors of $p$ in~$G_3$.
Let $N(p) \subseteq S(p)$ be the subset of these points that are locally reachable from $p$.
We connect $p$ to each point $q \in N(p)$ by an edge in $E_3$ of cost $\pi(p,q)$.
By definition, the minimal-cost path between $p$ and $q$ lies inside $T$.
Finally, as in $G_1$ and $G_2$, we add the portion of each edge of $\dvd$ between two sample
points as an edge of $E_3$.

\begin{figure}[t]
  \centering
   	\includegraphics[scale=.34]{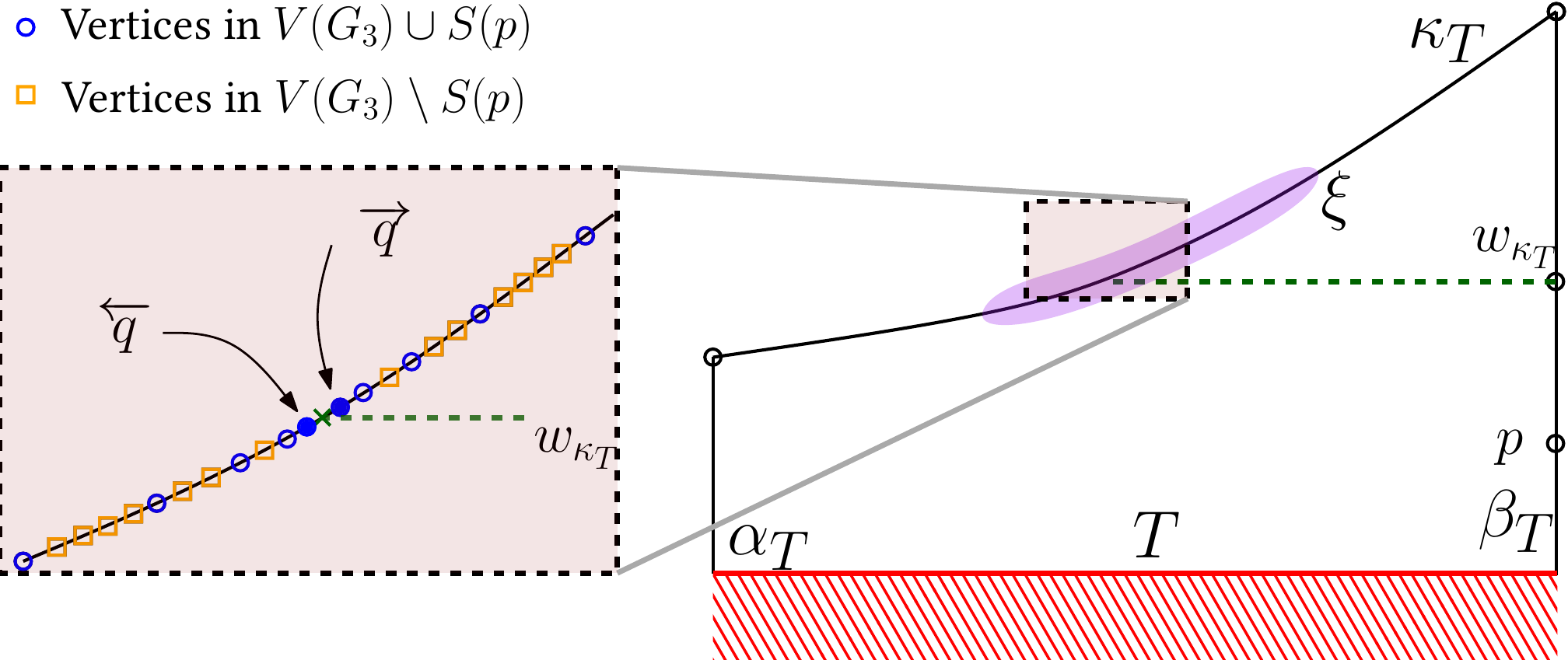}
	\caption{\sf
				Vertices in $S(p)$ which are used to construct the set of edges of~$G_3$.}
  \label{fig:edges}
\end{figure}
We now describe how we construct $S(p)$.
Let $\xi$ be an edgelet of $\partial T$ such that $p$ and $\xi$ do not lie on the same edge of
$T$.
We first define a \emph{shadow point} $\breve{p}$ of $p$.
If $p \in \alpha_T \cup \kappa_T$, then $\breve{p} = p$.
If $p \in \beta_T$, and $\xi \in \kappa_T$ (resp. $\xi \in \alpha_T$), then $\breve{p} = p$ if
$\cl(p) \geq \cl(w_{\kappa_T})$ (resp. $\cl(p) \geq \cl(w_{\alpha_T})$), and $\breve{p} =
w_{\kappa_T}$ (resp. $w_{\alpha_T}$) otherwise.
Let $\overleftarrow{q}$ (resp. $\overrightarrow{q}$) be the sample point on $\xi$ of highest
(resp. lowest) clearance less (resp. more) than $\breve{p}$, if such a point exists.
Exactly one of $\overleftarrow{q}$ or $\overrightarrow{q}$ may not exist if no point of clearance
$\breve{p}$ exists on $\xi$;
in this case, the construction implies that $\overleftarrow{q}$ is the endpoint of $\xi$ of higher
clearance or $\overrightarrow{q}$ is the endpoint of $\xi$ of lower clearance.
If $\overleftarrow{q}$ exists, we add $\overleftarrow{q}$ to $S(p)$.
We iteratively walk along sample points of $\xi$ in decreasing order of clearance starting with
the first sample point encountered after $\overleftarrow{q}$.
For each non-negative integer $i$, we add the point $q_i$ encountered at step
$\lfloor(1+\eps)^i\rfloor$ of the walk until we reach an endpoint of $\xi$.
Similarly, if $\overrightarrow{q}$ exists, we add to $S(p)$ the point $\overrightarrow{q}$ and
perform the walk along points of \emph{greater} clearance.
See Figure~\ref{fig:edges}.
Finally, we add the two endpoints of $\xi$ to $S(p)$.
We repeat this step for all edgelets on $\partial T$.
Since $\xi$ has $O(n/\eps)$ sample points and $\partial T$ has at most four edgelets, $|S(p)| =
O((1/\eps) \log (n / \eps))$, and $S(p)$ can be constructed in $O(|S(p)|)$ time.

\paragraph{Analysis.}
It is clear from the construction that $|V_3| = O(\frac{n^2}{\eps})$, $|E_3| =
O(\frac{n^2}{\eps^2} \log \frac{n}{\eps})$, and that $G_3$ can be constructed in time
$O(\frac{n^2}{\eps^2} \log \frac{n}{\eps})$.
By using Dijkstra's algorithm with Fibonacci heaps~\cite{FT87}, a minimal-cost path in~$G_3$ can
be computed in $O(\frac{n^2}{\eps^2} \log \frac{n}{\eps})$ time.
So it remains to prove that the algorithm returns a path of cost at most $(1 + O(\eps))\pi(s,t)$.
By rescaling $\eps$, we can thus compute a path from $s$ to $t$ of cost at most $(1+\eps)\pi(s,t)$
in time $O(\frac{n^2}{\eps^2} \log \frac{n}{\eps})$.

\begin{lemma}
  \label{lem:sample_bounds_ptas}
  Let~$\Pi^*$ be a minimal-cost path from $s$ to $t$.
  For every edge $e \in \dvd$, $\Pi^* \cap e$ lies inside the marked portion of $e$.
\end{lemma}
\begin{proof}
  Fix an edge $e \in \dvd$, and let $q \in \Pi^* \cap e$.
  We aim to prove $q$ lies inside the marked portion of $e$.
  Recall, $d \geq \OPT$.
  The proof of Lemma~\ref{lem:sample_bounds} already handles the case of $e$ being an internal
  edge.

  Now, suppose $e$ is an external edge.
  We assume $\mu(e) > 2d$; otherwise, the proof is trivial.
  We have $e \in \kappa_T$ and $e \in \kappa_{T'}$ for two adjacent Voronoi cells $T$ and $T'$.
  By construction, point $s$ lies outside the interior of $T \cup T'$.
  Therefore,~$\Pi^*[s,q]$ intersects at least one internal edge incident to~$T$ or~$T'$ at some
  point~$p$.
  \WLOG, assume that internal edge belongs to~$T$.
  We have two cases.
  \begin{enumerate}[align=left]
    \item[Case 1: $p \in \alpha_T$.]
      Since $\cl(p) \leq \cl(u_T) \leq \cl(q)$, we have $\pi(p, u_T) \leq \OPT$.
      By triangle inequality,
      $$ \pi(u_T, q) \leq \pi(u_T, p) + \pi(p, q) \leq 2\OPT \leq 2d. $$
    \item[Case 2: $p \in \beta_T$.]
      By~\eqref{eq:length_away}, $\pi(v^+_T, p) \leq 2d$.
      Recall from the proof of Lemma~\ref{lem:sample_bounds} that the clearance of any point
      on~$\Pi^*$ is at most $\cl(s) \cdot \exp(d)$.
      We defined $\eta_{v^+_T}$ as the line segment or circular arc with $v^+_T$ as one of its
      endpoints; $v'_T$ is its other endpoint.
      One can easily verify $\mu(\eta_{v^+_T}) \leq \OPT$; see the proof of
      Lemma~\ref{lem:hard_well-behaved_cost}.
      By the triangle inequality,
      $$ \pi(v'_T, q) \leq \mu(\eta_{v^+_T}) + \pi(v^+_T, p) + \pi(p, q) \leq 2\OPT + 2d \leq 4d.
      $$
  \end{enumerate}
\end{proof}

Now, we prove a property of locally reachable points from a fixed point which will be crucial for
our analysis.

\begin{figure}[t]
  \centering
   	\includegraphics[scale=.34]{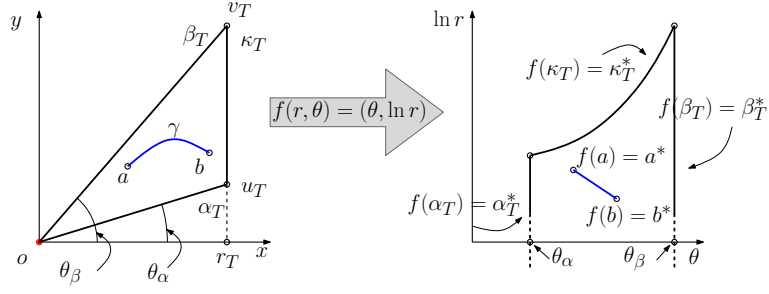}
    \caption{Case 1 of proof of Lemma~\ref{lem:visibility}:
   			a cell of a point obstacle in~\dvd and the optimal path between 
   			points $a$ and $b$ (blue) in the original (left) and transformed plane
        (right), respectively.}
    
   	\label{fig:transformed}
\end{figure}
\begin{figure}
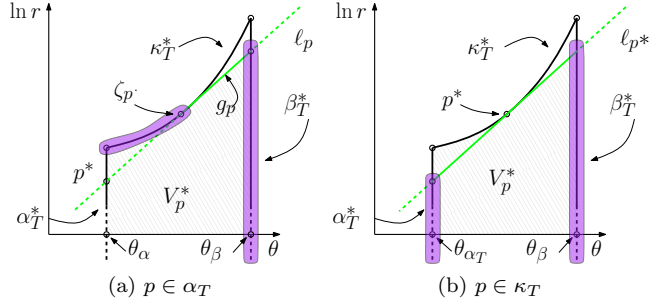

  \centering
  \subfloat
   [\sf $p \in \alpha_T$]
   {
   	\includegraphics[scale=.34]{lem_46_case_1a.pdf}
   	\label{fig:lem_46_case_1a}
   }
  \subfloat
   [\sf $p \in \kappa_T$]
   {
   	\includegraphics[scale=.34]{lem_46_case_1b.pdf}
   	\label{fig:lem_46_case_1b}
   }
   \caption{\sf Case 1 of proof of Lemma~\ref{lem:visibility}:
   			the set of locally reachable points (purple) from the point $p$ in
   			the transformed plane; the solid part of $\ell_p$ is $g_p$.}
  \label{fig:vis_case1}
\end{figure}

\begin{lemma}
\label{lem:visibility}
Let~$T$ be a cell of~$\dvd$, and let~$p \in \partial T$.
For every edge $e$ of $T$, the set of points on~$e$ locally reachable from~$p$, if non-empty, is a
connected portion of $e$ and contains an endpoint of $e$.
\end{lemma}
\begin{proof}
Let $o$ be the feature of $\calO$ associated with $T$.
We consider two main cases.
\begin{enumerate}[wide,itemindent=0pt]
  \item[Case 1: $o$ is a vertex.]
\WLOG,~$o$ lies at the origin, edge~$\alpha_T$ intersects the line~$y = 0$ at
the origin with angle~$\theta_{\alpha}$,
edge~$\beta_T$ intersects the line~$y = 0$ at the origin with angle~$\theta_{\beta}$,
and~$\theta_{\beta} > \theta_{\alpha} \geq 0$.
We consider a map $f : \R^2 \rightarrow \R^2 $ taking points to what we refer to as the
  \emph{transformed plane}. 
Given in polar coordinates point $(r, \theta)$, the map $f$ is defined as $f(r, \theta) = (\theta, \ln r)$.
For a point $x \in \R^2$, let $x^* = f(x)$, and for a point set $X \subseteq \R^2$, let $X^* =
\{f(x) | x \in X\}$.
Both~$\alpha_T$ and~$\beta_T$ become vertical rays in the transformed plane going to~$-\infty$.
Further, it is straightforward to show that~$\kappa_T$ becomes a convex curve in the transformed plane
when restricted to values of~$\theta$ such that~$\theta_{\alpha} \leq \theta \leq \theta_{\beta}$.
Therefore, $T^*$ is a semi-bounded pseduo-trapezoid.
By (P1) in Section~\ref{sec:alg_back} (see also~\cite{WBH08}), the minimal-cost path with respect
to $o$ between two points $a,b \in T$ maps to the line segment $a^* b^*$.
So $a$ and $b$ are locally reachable if $a^* b^* \subseteq T^*$, i.e., $a^*$ and $b^*$ are visible
from each other (see Figure~\ref{fig:vis_case1}).

For a point $p^* \in (\partial T)^*$, let $V_p^* \subseteq T^*$ be the set of points of $T^*$
visible from $p^*$, $\ell_p$ the line tangent to $\kappa_T^*$ from $p^*$ (if it exists), and
$\zeta_p = \ell_p \cap \kappa_T^*$.
Note that $\ell_p$ is well defined, because either~$p^* \in \kappa^*_T$ or the $x$-monotone convex
curve $\kappa^*_T$ either lies to the left or to the right of $p^*$.
The closure of $\partial V_p^* \setminus (\partial T)^*$ consists of a line segment $g_p = a^* b^*
\subset \ell_p$.
If $p \notin \kappa_T$, then $\zeta_p$ is one endpoint of $g_p$ and the other endpoint lies on
$\alpha_T^*$ or $\beta_T^*$.
In either case, for any edge $e \in \partial T$, if $(e^* \setminus \{p^*\}) \cap V_p^* \neq
\emptyset$, then it is a connected arc and contains one of the endpoints of $e^*$, as claimed.

\begin{figure*}[t]
  \centering
  \subfloat
   [\sf]
   { 
   	\includegraphics[scale=.6]{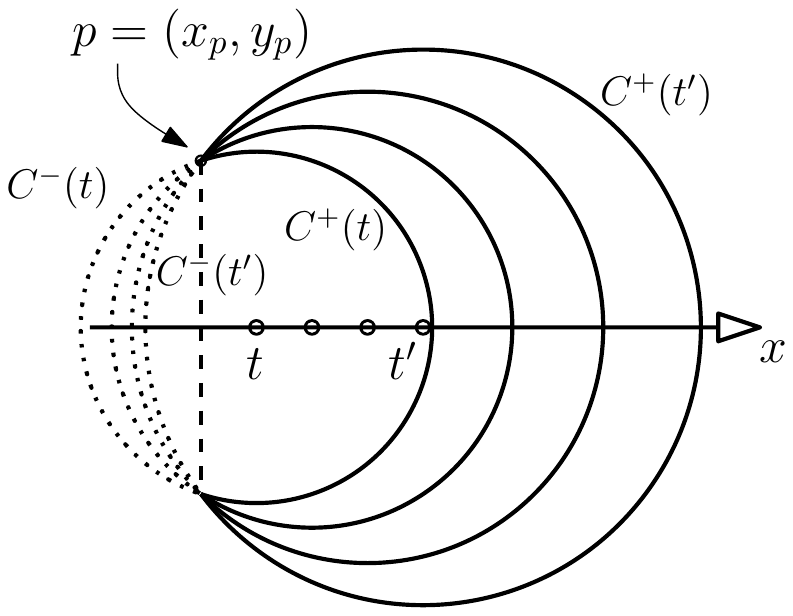}
   	\label{fig:Lem8_case2a}
   }
   \hspace{8mm}
  \subfloat
   [\sf]
   {
   	\includegraphics[scale=.34]{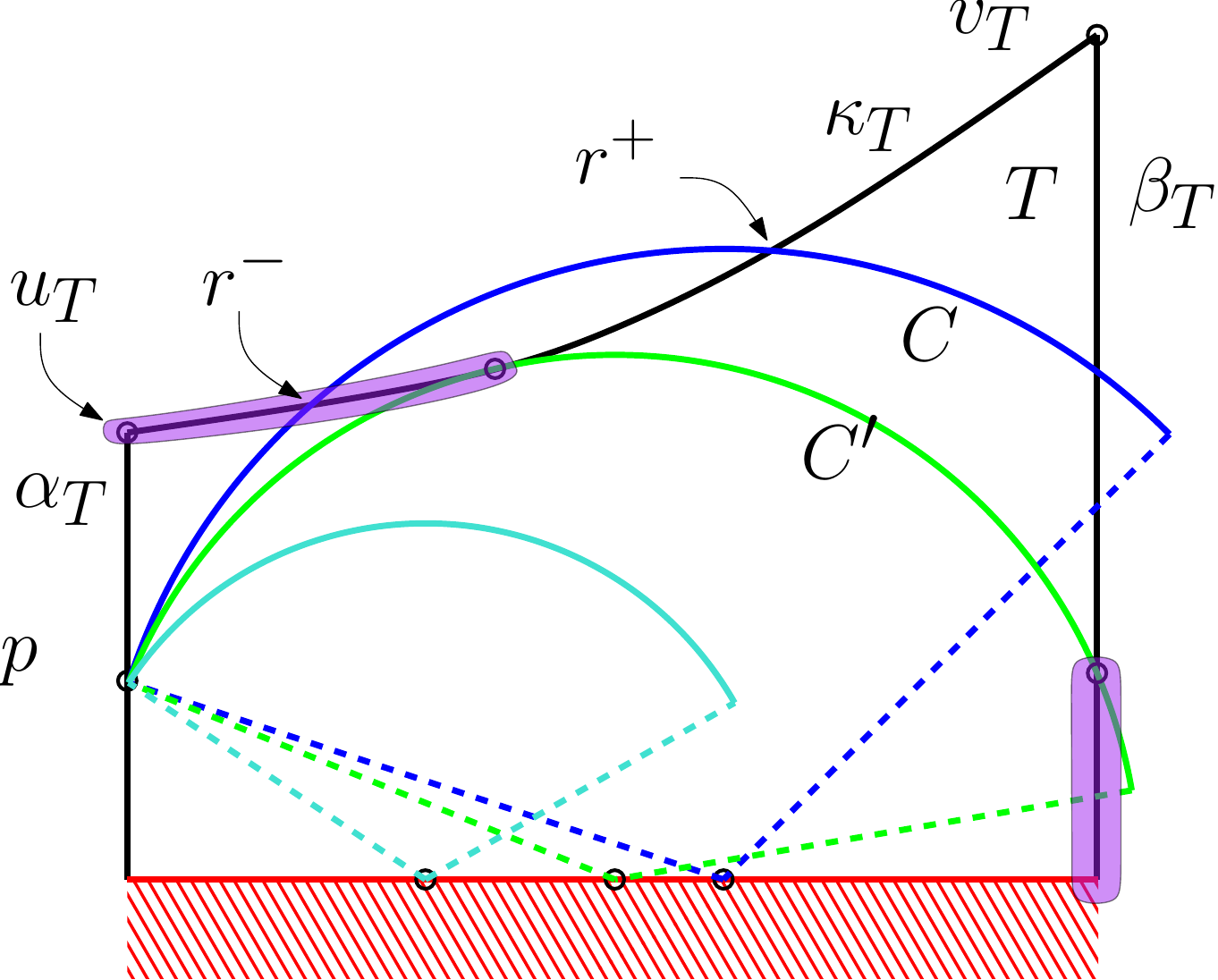}
   	\label{fig:Lem8_case2c}
   }
   \caption{\sf Illustration of properties of $\calC_p$ in Case 2 of proof of Lemma~\ref{lem:visibility}.}
\end{figure*}

\item[Case 2: $o$ is an edge.]
\WLOG, $o$ lies on the line~$y = 0$, the edge~$\alpha_T$ lies on the line~$x =
x_{\alpha}$, the edge~$\beta_T$ lies on the line~$x = x_{\beta}$, and~$x_{\beta} > x_{\alpha} \geq
0$.
There is no equally convenient notion of the transformed plane for edge feature~$o$, but we are
still able to use similar arguments to those given in Case 1.

In this case, for two points $a,b \in T$, the minimal-cost path with respect to $o$ from $a$ to
$b$ is the circular arc with $a$ and $b$ as its endpoints and centered at the $x$-axis (see (P2)
in Section~\ref{sec:alg_back}).
Therefore, $a$ and $b$ are locally reachable if this circular arc does not cross $\kappa_T$.

Fix a point $p = (x_p, y_p) \in \partial T$.
If $p \in \alpha_t \cup \beta_T$, then all points on the edge of $T$ containing $p$ are locally
reachable, and if $p \in \kappa_T$ then no point on $\kappa_T \setminus \{p\}$ is locally
reachable from $p$.
So we will focus on edges of $T$ that do not contain $p$.

Let $\calC_p$ denote the one-parameter family of circles that pass through $p$ and that are
centered at the $x$-axis.
For any $q \in T \setminus \{p\}$, there is a unique circle $C_q \in \calC_p$ that passes through
$q$.
We parameterize the circles in $\calC_p$ with the $x$-coordinate of its center, i.e., $C(t) \in
\calC_p$ for $t \in (-\infty, \infty)$ and is centered at $(t, 0)$.
Let $C^+(t)$ (resp. $C^-(t)$) be the circular arc of $C(t)$ lying to the right (resp. left) of the
line $x = x_p$.
See Figure~\subref*{fig:Lem8_case2a}.
The following properties of $\calC_p$ are easily verified:
\begin{enumerate}[label=(\alph*)]
  \item
    For $t < t'$, $C^+(t)$ (resp. $C^-(t')$) lies in the interior of $C(t')$ (resp. $C(t)$); see
    Figure~\subref*{fig:Lem8_case2a}.
  \item
    If a circle $C \in \calC_p$ intersects $\kappa_T$ at two points, say, $r^-$ and $r^+$,
    then there is another circle $C' \in \calC_p$ that is tangent to $\kappa_T$ between $r^-$
    and $r^+$; see Figure~\subref*{fig:Lem8_case2c}.
  \item
    A circle in $\calC_p$ intersects $\alpha_T$ or $\beta_T$ in at most one point.
  \item
    There is at most one circle $C \in \calC_p$ that is tangent to $\kappa_T$.
\end{enumerate}

Properties (a) and (b) are straightforward; (c) follows from (a) and a continuity argument; (d)
follows from (a), (c), and the convexity of $\kappa_T$.

If there is no circle in $\calC_p$ that is tangent to $\kappa_T$ then for any point $q \in T$, the
arc $C_q[p,q]$ lies inside $T$, so every point in $T$ is locally reachable, and the
lemma follows.

Next, assume there is a circle $C_0 \in \calC_p$ that is tangent to $\kappa_T$ at a point $r_p$.
By (d), $C_0$ is the only such circle.
There are three cases:
\begin{enumerate}[label=(\roman*)]
  \item
    If $p \notin \alpha_T$, then points in $\text{int}(C) \cap \alpha_T$ are locally reachable from
    $p$ by property (a).
  \item
    Similar, if $p \notin \beta_T$, then the points in $\text{int}(C) \cap \beta_T$ are locally
    reachable, again by property (a).
  \item
    If $p \in \alpha_T$ (resp. $p \in \beta_T$), then the points in $\kappa_T[u_t, r_p]$ (resp.
    $\kappa_T[r_p, v_T]$) are locally reachable from $p$ by properties (c) and (d).
\end{enumerate}

Hence, in each case at most one connected portion of an edge $e$ of $T$ is locally reachable from
$p$, and it contains one endpoint of $e$.
\end{enumerate}
\end{proof}

\begin{lemma}
\label{lem:eps-approximate}
Graph~$G_3$ contains an $s,t$-path of cost at 
most~$(1 + O(\eps))\OPT$.
\end{lemma}
\begin{proof}
Once again, we deform the optimal path~$\Pi^*$ into a path~$\tilde{\Pi}$ of~$G_3$ as in
Lemmas~\ref{lem:n-approximate} and~\ref{lem:constant-approximate}.
Let~$\Pi$ denote the current path that we have obtained by deforming~$\Pi^*$.
Let~$T \in \dvd$ be the first cell such that $\Pi$ enters $\interior(T)$ but $\interior(T) \cap
\Pi$ is not an arc of $E_3$.
Let $p \in \Pi$ be the first point (on $\partial T$) at which $\Pi$ enters in $\interior(T)$, and
let $q$ be the next point on $\Pi \cap \partial T$, i.e., $\interior(\Pi[p,q]) \subset
\interior(T)$.
If both $p$ and $q$ lie on the same edge~$e$ of $T$, we replace~$\Pi^*[p,q]$ with the portion
of~$e$ between~$p$ and~$q$, denoted by $\tilde{\Pi}_T$; note that $\mu(\tilde{\Pi}_T) = \pi(p,q)$.

Now, suppose $p \in \beta_T$ and $q \in \kappa_T$.
The other cases are similar.
Points~$p$ and~$q$ are locally reachable from each other.
By Lemmas~\ref{lem:sample_bounds_ptas} and~\ref{lem:visibility}, there exists a sample point~$p'$
locally reachable from~$q$ on~$\beta_T$ such that~$\pi(p,p') \leq \eps {d}/{n}$.
We have~$\pi(p',q) \leq \pi(p,q) + \eps {d}/{n}$.
Suppose there exists a point~$q' \in S(p')$ on~$\kappa_T$ locally reachable from~$p'$ such
that~$\pi(q,q') \leq \eps {d}/{n}$.
Let~$a$ be the minimal-cost path from~$p'$ to~$q'$.
In this case, we replace~$\Pi^*[p,q]$ with~$\tilde{\Pi}_T = pp' \circ a \circ \gamma(q', q)$.
We have $\mu(\tilde{\Pi}_T) \leq \pi(p,q) + 4 \eps {d}/n$.

Finally, suppose there is no locally reachable~$q'$ as described above.
As in Section~\ref{sec:prelim}, let~$\bar{w}^*$ denote the first intersection of well-behaved path
$\gamma(p',q)$ with~$\kappa_T$.
Recall our algorithm adds sample points along several edgelets of length~$O({d})$ such that each
pair of samples lies at cost~$\eps {d}/{n}$ apart.
By Lemma~\ref{lem:sample_bounds_ptas}, point~$q$ lies on one of these edgelets~$\xi$.

By Lemma~\ref{lem:good_anchors} and construction, either $\bar{w}^* \in \xi$ and~$\bar{w}^*$ lies between consecutive sample points of~$\xi$ we denoted as
$\overleftarrow{q}$ and $\overrightarrow{q}$, or $\bar{w}^* \notin \xi$ and exactly one of
$\overleftarrow{q}$ or $\overrightarrow{q}$ exists at an endpoint of $\xi$.
By construction, each existing point of $\overleftarrow{q}$ and $\overrightarrow{q}$ is in $S(p')$.
Let $q_- \in \{\overleftarrow{q}, \overrightarrow{q}\}$ be the first sample point of~$\xi$
encountered as we walk along $\kappa_T$ from $\bar{w}^*$, past~$q$, and to an endpoint of
$\kappa_T$.
We claim there exists at least one additional sample point of~$\xi$ other than $q_-$ encountered
during this walk, and we denote $q_0$ as the first of these sample points.
Indeed, if~$q_0$ does not exist, then $\bar{w}^* \in \xi$ and~$q$ lies between $\overleftarrow{q}$
and $\overrightarrow{q}$.
At least one of them is locally reachable from~$p'$ by Lemma~\ref{lem:visibility}, which
contradicts the assumption that~$q$ is at least~$\eps {d}/{n}$ cost away from any sample point
of~$\xi \cap S(p')$ locally reachable from~$p'$.
By a similar argument, we claim $q$ does not lie between $q_0$ and $\bar{w}^*$.

\begin{figure}[t]
  \centering
   	\includegraphics[scale=.34]{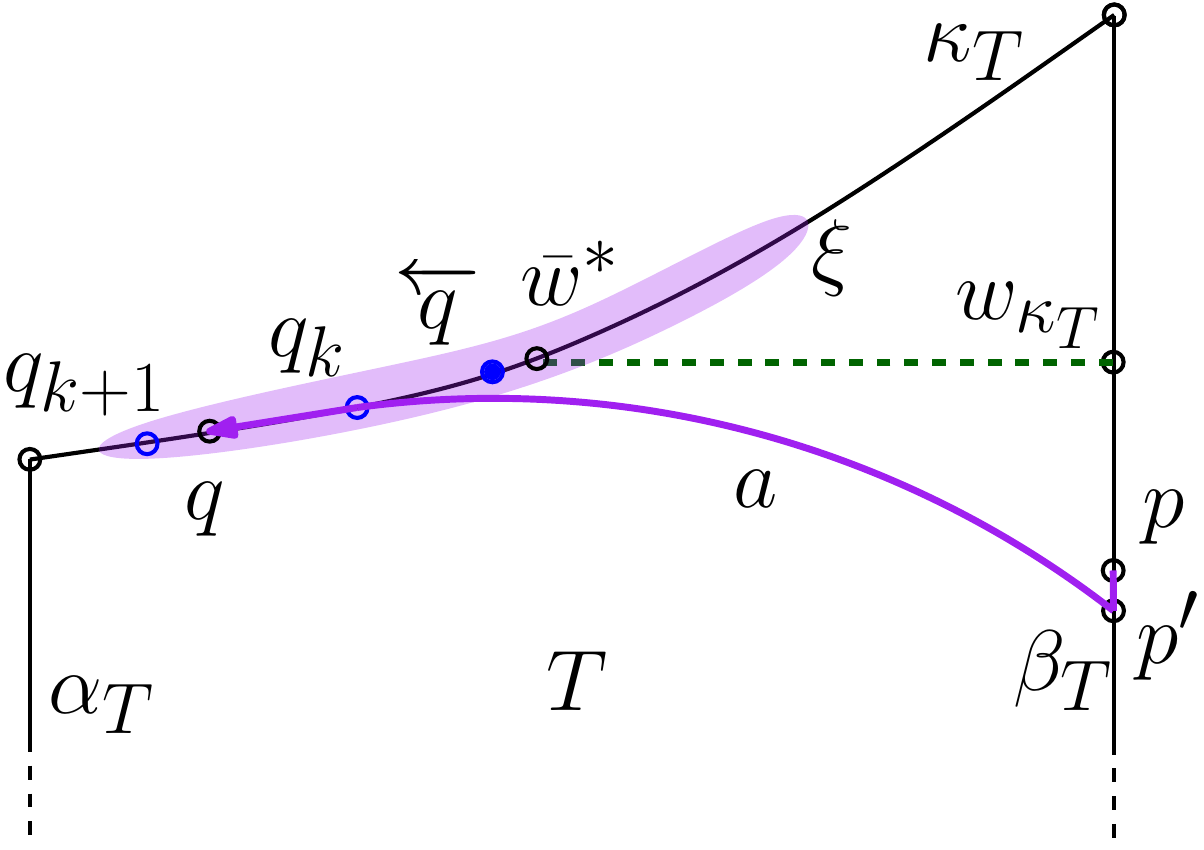}
	\caption{\sf Visualization of proof of Lemma~\ref{lem:eps-approximate}.}
  \label{fig:ptas_proof}
\end{figure}
Recall, our algorithm adds samples~$q_i$ to~$S(p)$ spaced geometrically away from one of
$\overleftarrow{q}$ and $\overrightarrow{q}$ in the direction of~$q$; point $q_0$ is one of these
samples.
These samples also include one endpoint of~$\xi$.
Let $q_k, q_{k+1}$ be two consecutive sample points of $S(p)$ such that $q$ lies between them.
By Lemma~\ref{lem:visibility}, at least one of~$q_k$ and~$q_{k+1}$ is locally reachable from~$p'$.
Let~$q'$ be this locally reachable point.
Let~$a$ be the mimimal-cost path from~$p'$ to~$q'$.
As before, we replace~$\Pi^*[p,q]$ with~$\tilde{\Pi}_T = pp' \circ a \circ \gamma(q', q)$.
See Figure~\ref{fig:ptas_proof}.

Let~$\delta = \pi(q_-, q) n/(\eps {d})$.
Value~$\delta$ is an upper bound on the number of samples in~$\xi$ between $q_-$ and~$q$.
We have~$\lfloor(1+\eps)^{k}\rfloor \leq \delta \leq \lfloor(1+\eps)^{k+1}\rfloor$.
In particular~$\delta \leq (1+\eps)^{k+1}$, which implies $\delta - \lfloor(1+\eps)^{k}\rfloor
\leq \eps \delta + 1$.
Similarly, $\lfloor(1+\eps)^{k+1}\rfloor - \delta \leq \eps \delta$.
By Lemma~\ref{lem:hard_well-behaved_cost},~$\pi(q_-, q) \leq 11 \mu(p', q)$.
We have
\begin{align*}
  \pi(q,q') &\leq (\eps \delta + 1) \frac{\eps {d}}{n}\\
            &\leq \left(\pi(q_-, q) \frac{\eps n}{\eps {d}} + 1\right)\frac{\eps {d}}{n}\\
            &= \eps \pi(q_-,q) + \frac{\eps {d}}{n}\\
            &\leq 11\eps \pi(p', q) + \frac{\eps {d}}{n}.
\end{align*}

We have~$\pi(p', q') \leq \pi(p', q) + \pi(q, q') \leq (1 + 11\eps)\cdot \pi(p',q) + \eps
{d}/{n}$.
Therefore, in all three cases we have
$$\mu(\tilde{\Pi}_T) \leq (1 + O(\eps)) \pi(p,q) + O(\eps {d} / n).$$

Summing over all steps in the deformation of $\Pi^*$ and using the fact $\OPT \leq d \leq c \OPT$
for a constant $c$, we obtain $\mu(\tilde{\Pi}) = (1 + O(\eps)) \OPT$.
As before, it is clear from the construction that $\tilde{\Pi}$ is a path in $G_3$.

%
%
%
%
\end{proof}

We conclude with our main theorem.

\begin{theorem}
\label{thm:ALG}
	Let~$\calO$ be a set of polygonal obstacles in the plane with~$n$ vertices total, and let~$s,t$ be two points outside~$\calO$.
  Given a parameter~$\eps \in (0,1]$,
	there exists an 
  $O(\frac{n^2}{\eps^2} \log \frac{n}{\eps})$-time approximation algorithm 
	for the minimal-cost path problem between~$s$ and~$t$
	such that the algorithm returns an $s,t$-path of cost at most~$(1 + \eps)\pi(s,t)$.
\end{theorem}

\section{Discussion}
In this paper, we present the first polynomial-time approximation algorithm for the problem of computing minimal-cost paths between two given points (when using the cost defined in~\eqref{eq:length}).
One immediate open problem is to improve the running time of our algorithm to be near-linear.
A possible approach would be to refine the notion of anchor points so it suffices to put only
$O(\log n)$ additional points on each edge of the refined Voronoi diagram.

Finally, there are other natural interesting open problems that we believe should be addressed.
The first is to determine if the problem at hand is NP-hard.
When considering the complexity of such a problem, one needs to consider both the \emph{algebraic
complexity} and the \emph{combinatorial complexity}.
In this case we suspect that the algebraic complexity may be high because of the cost function we consider.
However, we believe that combinatorial complexity, defined analogously to the number of ``edge
sequences'', may be small.
The second natural interesting open problem calls for extending our algorithm to compute near-optimal paths amid polyhedral obstacles in~$\R^3$.

\bibliographystyle{abuser}
\bibliography{bibliography}
\end{document}